\documentclass[aps,prl,superscriptaddress,twocolumn
,notitlepage,longbibliography
]
{revtex4-2}
\usepackage{comment}

\usepackage{dcolumn}
\usepackage{bm}
\usepackage[dvipsnames]{xcolor}
\usepackage[mathlines]{lineno}
\usepackage{amsthm,amssymb,amsmath}
\usepackage{physics}
\usepackage{mathrsfs}
\usepackage{mathtools}
\usepackage{xcolor, soul}
\usepackage{tikz}
\usepackage{epic}
\usepackage[percent]{overpic} 
\rmfamily
\usepackage{gensymb}
\usepackage[colorlinks]{hyperref}
\usepackage{floatrow}
\usepackage[caption=false]{subfig}
\definecolor{myorange}{RGB}{255,165,0}
\hypersetup{
 linkcolor=NavyBlue
,citecolor=Red
,urlcolor=NavyBlue
}

\definecolor{bur}{rgb}{0.5, 0.0, 0.13}
\definecolor{ngreen}{rgb}{0.5, 0.6, 0.35}

\newcommand\blk[1]{\color{black}#1}

\newcommand\mineig[1]{\xi_{\mathrm{min}} \left(#1\right)}
\newcommand\maxeig[1]{\xi_{\mathrm{max}} \left(#1\right)}

\newcommand{\eref}{Eq.~\eqref}

\newtheorem{theorem}{Theorem}
\newtheorem{lemma}[theorem]{Lemma}

\graphicspath{{images/}}

\begin{document}
\title{
Detection-loophole-free nonlocality in the simplest scenario
}

\author{Nandana T Raveendranath}%
\affiliation{Quantum and Advanced Technologies Research Institute, Griffith University, Yuggera Country, Brisbane, QLD 4111, Australia}
\author{Travis J. Baker}
\email{dr.travis.j.baker@gmail.com}
\affiliation{Quantum and Advanced Technologies Research Institute, Griffith University, Yuggera Country, Brisbane, QLD 4111, Australia}
\affiliation{%
Nanyang Quantum Hub, School of Physical and Mathematical Sciences, Nanyang Technological
University, Singapore 637371
}%
\author{Emanuele Polino} \email{e.polino@griffith.edu.au}
\author{Marwan Haddara}%
\affiliation{Quantum and Advanced Technologies Research Institute, Griffith University, Yuggera Country, Brisbane, QLD 4111, Australia}
\author{Lynden K. Shalm}
\author{Varun B. Verma}
\affiliation{National Institute of Standards and Technology, 325 Broadway, Boulder, Colorado 80305, USA}
\author{Geoff J. Pryde}%
\author{Sergei Slussarenko}%
\author{Howard M. Wiseman}%
\author{Nora Tischler}%
\email{n.tischler@griffith.edu.au}
\affiliation{Quantum and Advanced Technologies Research Institute, Griffith University, Yuggera Country, Brisbane, QLD 4111, Australia}

\begin{abstract}

Loophole-free quantum nonlocality often demands experiments
with high complexity (defined by all parties’ settings and outcomes) and multiple efficient detectors. 
Here, we identify the fundamental efficiency and complexity thresholds for quantum steering using two-qubit entangled states. 
Remarkably, it requires only one photon detector on the untrusted side, with efficiency $\epsilon > 1/X$, where $X\geq2$ is the number of settings on that side. 
This threshold applies to all pure entangled states, in contrast to analogous Bell-nonlocality tests, which require almost unentangled states to be loss-tolerant. 
We confirm these predictions in a minimal-complexity ($X = 2$ for the untrusted party and a single three-outcome measurement
for the trusted party), detection-loophole-free photonic experiment with $\epsilon = (51.6 \pm 0.4)\%$.
\end{abstract}
\maketitle

Quantum nonlocality plays a key role in the foundations of quantum physics and is an essential resource for emerging quantum information science applications~\cite{RevModPhys.89.015004,RevModPhys.94.025008,branciard2012one,PhysRevLett.97.120405,comandar2016quantum,zeng2025steering,joch2022certified,li2024randomness,kavuriTraceableRandomNumbers2025a}. 
Observing behaviours generated by entangled quantum states requires the parties that share the state, say Alice and Bob, to perform suitably chosen local measurements, each with two or more possible outcomes.
To rigorously verify nonlocal correlations, one cannot trust the measurement devices of all parties. 
This lack of trust entails a detection loophole~\cite{larsson2014loopholes,gisin1999local} 
that must be closed to prevent any untrusted party from mimicking the statistics of nonlocal correlations without entanglement. 
Closing loopholes is essential for ensuring genuine nonlocality~\cite{christensenDetectionLoopholeFreeTestQuantum2013,shalmStrongLoopholeFreeTest2015,PhysRevLett.115.250401, Hen15,smith2012conclusive,wittmann2012loophole,fuwa2015experimental,guerreiro2016demonstration,cavailles2018demonstration,srivastav2022quick,slussarenko2022quantum,zeng2025steering}, which is crucial for device-independent applications~\cite{RevModPhys.89.015004,RevModPhys.94.025008,branciard2012one, PhysRevLett.97.120405,comandar2016quantum,zeng2025steering,joch2022certified,li2024randomness,kavuriTraceableRandomNumbers2025a}.

Key resources for detection-loophole-free nonlocality demonstrations are: (1) experimental complexity, quantified by the number of possible detection patterns (settings and outcomes of all parties together), and (2) the minimum required detection efficiency.
If Alice (Bob) has $X\;(Y)$ settings, each with $A\;(B)$ possible outcomes,
the complexity cost \cite{Saunders2012} is 
\begin{equation}{\label{Complexity_cost}}
W \coloneqq A^X B^Y \;.
\end{equation}
Loophole-free Bell-nonlocality \cite{brunner2014bell}, where both the parties are untrusted, has been demonstrated with a minimum complexity cost of $W=16$ (two parties, two settings, and two outcomes per party) \cite{shalmStrongLoopholeFreeTest2015, Hen15, PhysRevLett.115.250401}. Of these,  the photonic experiments \cite{shalmStrongLoopholeFreeTest2015,PhysRevLett.115.250401} made do with $A=B=2$, by assigning null results deterministically to one of the non-null results. 
This allowed them to use only a {\em single} high-efficiency photon detector on each side (efficiency $>2/3$~\cite{Eberhard1993}). However, this strategy comes at the expense of using non-maximally entangled states, which are more noise-sensitive. 

Yet, nonlocality can be observed at a lower complexity cost~\cite{Saunders2012}, using quantum steering \cite{schrodinger1935discussion}, 
also known as EPR-steering after Einstein, Podolsky, and Rosen~\cite{EPR35}. 
Quantum steering as now formalized~\cite{Wis07,Rei09,Uol20,Cav09,Cav16} is a form of nonlocality requiring trust to be placed on one party, who uses well-characterized quantum devices to describe their local system.
The untrusted party can remotely prepare different ensembles received by the trusted party using a choice of distinct measurement settings~\cite{EPR35,Wis07}. 
The minimal-cost quantum steering test, requiring only two dichotomic measurements for the untrusted party and a single three-outcome measurement for the trusted party ($W=12$), was proposed in Ref.~\cite{Saunders2012}. 
However, the experimental demonstration in Ref.~\cite{Saunders2012} ignored non-detection events, and therefore left the detection loophole open. 
On the other hand, steering without this fair sampling assumption has been performed~\cite{bennet2012arbitrarily,wittmann2012loophole,smith2012conclusive}, 
using two photon detectors with heralding  efficiencies exceeding the threshold of $1/X$ on the untrusted side~\cite{bennet2012arbitrarily}. However, these loss-tolerant tests did not minimize the complexity cost; the inclusion of a distinct ``null'' measurement outcome (when neither of the untrusted detectors click)~\cite{bennet2012arbitrarily,evans2013loss} increases $A$ in \eref{Complexity_cost} from 2 to 3. The efficiency thresholds for minimal-complexity steering have, until now, been unknown.

Here, we determine these fundamental efficiency bounds and identify the entanglement resources required to attain them.
Specifically, we find that the efficiency thresholds of $1/X$ for loss-tolerant photonic quantum steering tests can be attained without increasing the complexity cost. 
This complexity cost is achieved via an experimentally convenient scenario where the untrusted party requires only a single detector. 
The trusted party has three detectors, which can have arbitrarily low efficiency as usual. 
In particular, the minimum complexity cost of $W=12$ can be achieved while closing the detection loophole as long as the one-photon detector used by Alice has an efficiency above $50\%$.
We perform an experiment demonstrating detection-loophole-free quantum steering in the simplest possible scenario. 
Notably, we achieve steering close to the minimum detection threshold, with an efficiency as low as $51.6 \pm 0.4\%$.
A novel family of one-detector loss-tolerant steering inequalities is derived, which permits demonstrations arbitrarily close to the fundamental efficiency bound in the simplest scenario.
Moreover, we show that the one-detector efficiency bound of $1/X$ can be attained for \emph{any} pure entangled state shared by Alice and Bob (and thus for entangled qubits in particular).
This property stands in stark contrast to Bell-nonlocality, where almost unentangled states are necessary to achieve the fundamental (Eberhard) efficiency bound of $2/3$, in the analogous minimal complexity test~\cite{Eberhard1993}. 
 
\emph{Quantum steering with one detector on the untrusted side.---}  
We consider a two-party steering scenario under the restriction that the untrusted party (Alice) has access to only \emph{one} detector, 
which projects her system into a rank-one subspace with efficiency $\epsilon$.
Each measurement Alice can perform, labelled $x=0,\dots,X-1$ can be referred to as a \emph{one-click} measurement, since only the ``click'' outcome, defined as $a=+$, is registered as a detection event by Alice.
All non-detection events, including system losses into any other modes not coupled to the detector, are labelled as null outcome, $a=\emptyset$.

Our goal is to decide whether quantum steering tests are possible with one-click measurements, given $\epsilon$.
To this end, suppose Alice and Bob share a pure bipartite state $\ket{\Psi_{\mathrm{AB}}}$, with local system dimensions $d_\mathrm{A}$ and $d_\mathrm{B}$.
The set of Bob's unnormalized states conditioned on Alice's outcome $a$ and setting $x$, $\{ \sigma_{a|x} \}_{a,x}$ (the subscript meaning the set ranges over all $a$ and $x$), is commonly called an assemblage. 
Such an assemblage demonstrates quantum steering if it precludes the existence of a local-hidden-state (LHS) model. This LHS model consists of an ensemble of quantum states $\{p_\lambda,\rho_\lambda\}_\lambda$ and probability distributions over $a$, $\{ P_\lambda(a|x)\}_{a,x}$ such that $\sigma_{a|x} = \sum_\lambda p_\lambda P_\lambda(a|x) \rho_\lambda \ \forall \ ~a,~x$ \cite{Wis07}.

In Section \ref{sm:1/X_bound} of the Supplemental Material (SM), we prove the following facts about assemblages prepared by one-click measurements on Alice's system.
First, by choosing her one-click effect to be proportional to a rank-one projector, 
$E_{+|x} = \epsilon~\Pi_{+|x}$, a necessary and sufficient condition for Alice to steer Bob is that her detector efficiency surpasses the threshold
\begin{equation}
\epsilon > \left[\maxeig{\sum_x \Pi_{+|x}}\right]^{-1}~,
\label{spectral_rad}
\end{equation}
where $\maxeig{\cdot}$ denotes maximum eigenvalue. 
This lower bound can be expressed in terms of the number of measurements performed by the untrusted party, $X$, and all pairwise overlaps of the projectors, 
$\Tr(\Pi_{+|x}\Pi_{+|x'})$; see Lemma~2 in the SM. 
The measurement overlap quantifies how distinct Alice's click-conditioned measurement effects are.
From this observation, we compute (see Theorem~3 in the SM) the infimum of the 
lower bound in Eq.~(\ref{spectral_rad}) to be $1/X$. 
That is, there exists a one-click measurement strategy defined by exactly $X$ projectors $\{ \Pi_{+|x}\}$, so that steering is always possible for 
\begin{equation} \label{eq:cutoff_efficiency}
\epsilon > \frac{1}{X} \;.
\end{equation}
Since steering is known to be impossible (regardless of the complexity of setup) for $\epsilon\leq 1/X$~\cite{bennet2012arbitrarily}, this defines a fundamental limit on efficiency thresholds required to show quantum steering with one-click measurements on Alice's side. 
This may surprise, since the proof in Ref.~\cite{bennet2012arbitrarily} that permits closing the detection loophole above the $1/X$ bound required two detectors---or equivalently two non-null measurement outcomes---at the untrusted side.
Interestingly, Eq.~\eqref{eq:cutoff_efficiency} is approached from above by taking the limit where the $\Pi_{+|x}$ 
converge, with $\Tr(\Pi_{+|x}\Pi_{+|x'})\to 1 \ \forall \ \ x,\ x'$. 
This is contrary to the intuition that the power of steering is maximized by Alice measuring maximally different observables~\cite{Eva14b, Bav17, designolle2019incompatibility}, but it aligns with Eberhard's result for maximal loss-tolerance in Bell-nonlocality~\cite{Eberhard1993}, as we discuss in a later Section. 
Furthermore, our measurement construction does not require Bob to perform complete tomography of his steered states; in fact, he needs only to implement a single three-outcome measurement.
Together, these results imply a striking conclusion: by choosing $X=2$, quantum steering can be made loophole-free, with just one detector above the fundamental efficiency threshold of $1/2$, and with a complexity cost of $W=12$, as low as when making the fair sampling assumption (as in Ref.~\cite{Saunders2012}).

To demonstrate steering at minimal complexity ($X=2$ one-click measurements), we analytically derive the corresponding family of optimal steering witnesses. 
For a given assemblage $\{\sigma_{a|x}\}_{a,x}$, a steering witness is specified by a set of Hermitian operators $\{ F_{a|x} \}_{a,x}$ such that
\begin{equation}
\sum\limits_{a,x} \Tr \left( F_{a|x} \sigma_{a|x}\right) \geq 0 \;
\label{eq:SE}
\end{equation}
for all non-steerable assemblages, with any violation certifying quantum steering~\cite{Pus13, Cav16,skrzypczyk2023semidefinite}.
In Section~\ref{sm:witness_derivation} of the SM, we construct an explicit family of witnesses $\{F^\star_{a,x} \}_{a,x}$, whose form depends on the spectrum of Bob's reduced state, the detector efficiency $\epsilon$, and the overlap between Alice's two ``click'' effects.
This analytic construction [see Eq.~\eqref{eq:SE_val_exact} in the SM] asymptotically approaches the efficiency bound of Eq.~\eqref{eq:cutoff_efficiency} as $\Tr(\Pi_{+|0}\Pi_{+|1})\to 1$, allowing violation for any pure entangled state with $\epsilon >1/2$.

\emph{Experimental implementation.---}
We experimentally implement the above simplest protocol for detection-loophole-free quantum steering with photons.
Our setup consists of three parts: the source, the untrusted party (Alice), and the trusted party (Bob), as depicted in Fig.~\ref{fig:experimental schematic}. A high-efficiency source of entangled states, following the approach of \cite{Tis18, shalmStrongLoopholeFreeTest2015}, is implemented to obtain polarization-entangled photon pairs at 1550 nm. 
The target state has the form
\begin{equation}\label{state}
\ket{\Phi^{+}_\alpha} = \cos(\alpha)\ket{HH} + \sin(\alpha)\ket{VV},
\end{equation}
where $|H \rangle$ and $|V \rangle$ represent horizontal and vertical polarizations, respectively, and $\alpha\in[0,\pi/4]$ is the tunable parameter determining the amount of entanglement in the state.
Our experimentally generated state has a fidelity of $0.9953 \pm 0.0006$ with the maximally entangled state when we set $\alpha = \pi/4$.
The two photons are sent to two parties, Alice and Bob, for measurement.
Bob performs a trine measurement with a fixed three-outcome POVM \cite{clarkeExperimentalRealizationOptimal2001,Saunders2012}.
Since Bob's measurements are trusted, the detection events associated with any outcomes on his side herald the trial, requiring Alice to announce her outcome for the trial.
Alice's outcomes are either ``click", when she detects a photon, or ``null", when she does not. 
The single detector on Alice's side performs these one-click measurements defined by the projectors $\Pi^\theta_x = \ketbra{e^\theta_x}$, with 
$\ket{e^\theta_x} \coloneqq \cos(\theta) \ket{H} + (-1)^x\sin(\theta) \ket{V}$,
for some $\theta$ implemented by the motorized half-wave plate (HWP) rotated by an angle $\theta/2$.   
Alice's detector needs to be
moderately efficient so that Alice's detection efficiency surpasses the threshold of 0.5, whereas Bob's detectors can be arbitrarily inefficient since he is trusted.

\begin{figure}[htbp]
\begin{center}
\includegraphics[width=\linewidth]{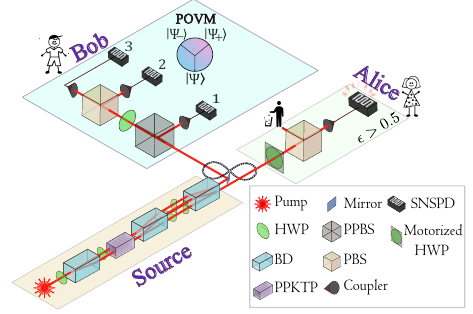}
\caption 
{\footnotesize{\textbf{Experimental setup.} The source generates 1550 nm polarization-entangled photons via SPDC (spontaneous parametric down-conversion) in a nonlinear crystal (PPKTP) embedded in a Mach-Zehnder interferometer realized by beam displacers (BDs). The entangled photons are sent to Alice and Bob. Alice performs a one-click measurement using a half-waveplate (HWP), a polarizing beam splitter (PBS), and only one detector. Bob's three-outcome POVM is implemented using a partially polarizing beam splitter (PPBS), a HWP, and a PBS with outcomes corresponding to polarizations in the X-Z plane of the Bloch sphere. }}
\label{fig:experimental schematic}
\end{center}
\end{figure}

\begin{figure}[htbp]
\begin{center}
\includegraphics[width= \linewidth]{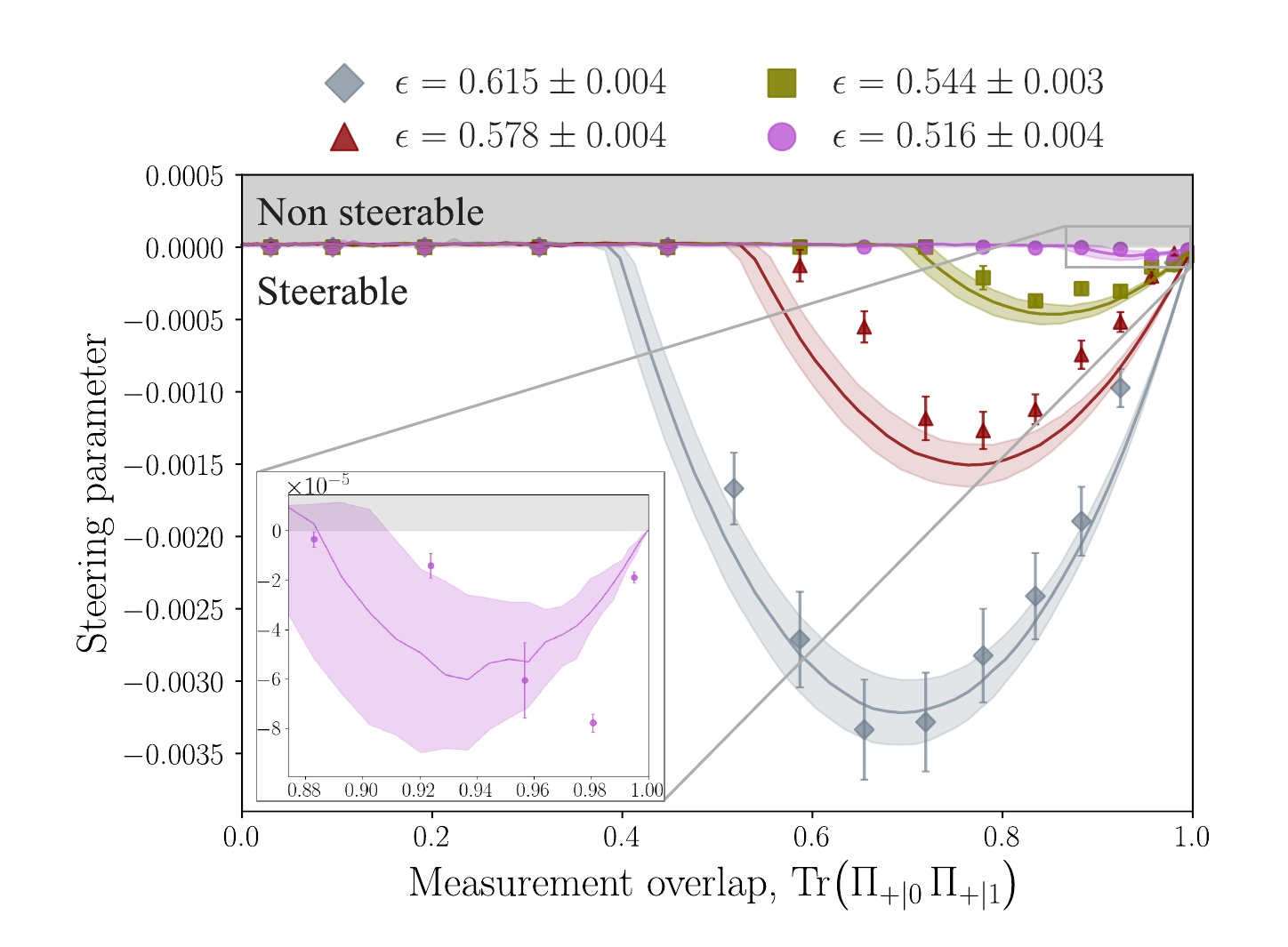}
\caption{\footnotesize 
Steering parameters versus Alice's measurement overlap for maximally entangled states. 
Points in the white region indicate a steering violation. 
The solid curves show theoretical predictions for an ideal state, Alice's projectors and Bob's POVM, based on the experimentally measured efficiencies; the shaded regions represent the uncertainty for those predictions based on the uncertainty in efficiency. 
Markers indicate the experimental data, along with the associated error estimated as $\pm1$ standard deviation, obtained by repeating the measurements 10 times (The values of minimum steering parameters and their errors are provided in SM Table \ref{tab:violation values}). 
As the efficiency decreases, Alice's measurements are required to have a high overlap to steer Bob. The inset zooms in on the lowest efficiency curve, with $\epsilon = 0.516 \pm 0.004$.
}
\label{balanced}
\end{center}
\end{figure}

From the experimental data, we reconstruct each element $\sigma_{a|x}$ of Bob's assemblage via a constrained maximum-likelihood estimation technique (see SM Section \ref{sm:experiment_details}), and 
compute the left side of Eq.~\eqref{eq:SE}.
We refer to this value as the steering parameter---a negative value unambiguously certifies quantum steering.
Steering tests are repeated for different detection efficiencies using a maximally entangled state ($\alpha=\pi / 4$ in Eq.~\eqref{state}). Both theoretically predicted and experimentally obtained steering parameters are shown as a function of the overlap between measurement settings of the untrusted party in Fig.~\ref{balanced}. 
As we approach lower efficiency values toward the ultimate threshold, 
we see that higher measurement overlaps are required to witness steering. 
We observe a steering parameter of $(-7.79 \pm 0.37) \times 10^{-5}$, violating the non-steerability bound by more than 21 standard deviations, for an efficiency of $\epsilon = 0.516 \pm 0.004$ (see inset of Fig.~\ref{balanced}). This efficiency is very close to the minimum efficiency bound, which shows the robustness of our protocol. 
We also report similar experimental violations using less entangled states, by taking $\alpha \rightarrow0$ in Eq.~\eqref{state}, but defer discussion to Section \ref{sm:experiment_details} of the SM.

\begin{figure}[htbp]
\begin{center}
\includegraphics[width=\linewidth]{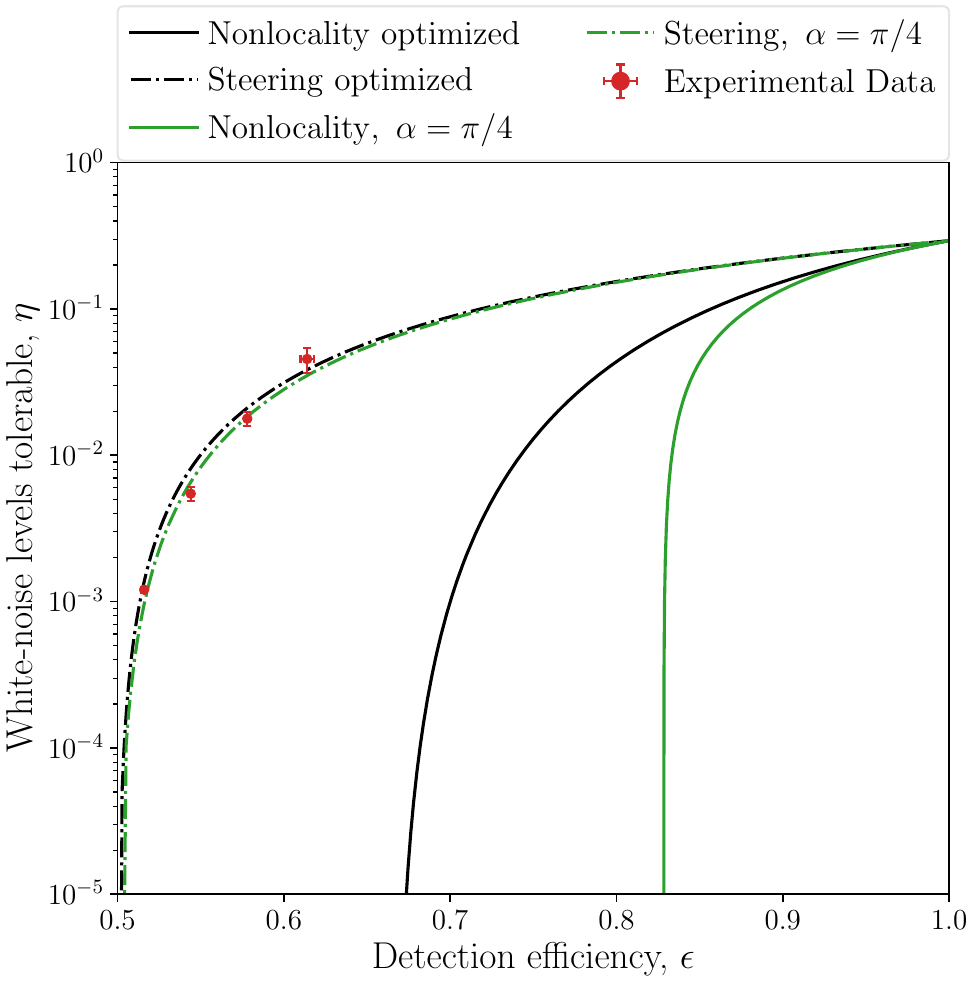}
\caption{\footnotesize 
White-noise levels  ($\eta$) tolerable for the simplest bipartite demonstrations of nonlocality, as a function of the detection efficiency $\epsilon$. 
Solid curves are for Bell-nonlocality, and dashed curves for EPR-steering. 
Green is for maximally entangled states and black for optimal states. Red points are experimentally measured white-noise robustness of steering with
1 standard deviation error bars. For details, see text. }
\label{fig:Eberhard plot}
\end{center}
\end{figure}

\emph{Entanglement and noise robustness in steering.---}
Eberhard’s seminal work~\cite{Eberhard1993} revealed that the most loss-tolerant demonstrations of Bell-nonlocality require an efficiency $\epsilon > 2/3$, and can be implemented with $A=B=X=Y=2$, giving $W=16$ in~\eref{Complexity_cost}.
In SM Sec.~\ref{sec:eberhard_overview} we provide an overview of Eberhard's work, and independently prove the $2/3$ efficiency threshold from below, using a more general black-box (theory-independent) framework.

Approaching Eberhard's efficiency threshold in Bell tests is possible only in the limit where the entanglement of the state Alice and Bob share approaches zero [$\alpha\rightarrow 0$ in Eq.~\eqref{state}], \emph{and} the overlap between the ``click'' projectors approaches unity for both Alice's and Bob's measurements. 
By contrast, if maximally entangled states are used, the efficiency threshold rises to $\epsilon > 2(\sqrt{2}-1) \approx 0.8284$.
This corresponds to the well-known \emph{anomaly} of Bell-nonlocality, first noted by Eberhard~\cite{Eberhard1993}, where the detector-efficiency threshold for closing the detection loophole decreases as the shared entanglement is reduced~\cite{methot2007anomaly}.

We now extend Eberhard’s analysis to establish, for quantum steering in its simplest configuration, the noise levels these nonlocality tests can tolerate, and the entanglement they require.
The white-noise robustness (WNR) is defined as the maximum fraction of white noise $\eta\geq0$ permissible in the entangled state, such that a Bell or steering~\cite{designolle2019incompatibility} inequality violation is possible.
In our analysis, we maximize such WNR over all one-click measurements numerically, and we obtain the results in Fig.~\ref{fig:Eberhard plot}.
The solid lines reproduce Eberhard’s noise-efficiency trade-offs for Bell-nonlocality \cite{Eberhard1993}, when measuring maximally entangled states (green), and when both the states and measurements are optimized (black); see SM Section~\ref{sec:eberhard_overview} for details. 
Similarly, Fig.~\ref{fig:Eberhard plot} shows numerically optimized trade-offs for steering restricted to maximally entangled states (dashed green curve), and optimized over all states and measurements (dashed black curve). 
For comparison, we estimate the WNR of the four experimental points maximally violating a steering inequality in Fig.~\ref{balanced}, and display these values as red points in Fig.~\ref{fig:Eberhard plot}, in close agreement with the corresponding theoretical prediction (green dashed curve).
We refer the reader to Section~\ref{sec:noisy_assemblage_witness} of the SM for details.

The most salient feature of Fig.~\ref{fig:Eberhard plot} is that the simplest quantum steering can be performed in regimes of noise and detection-efficiency that are inaccessible to the simplest Bell-nonlocality. Another clear difference is that  
the maximally entangled states (green) attain close-to-maximal (black) noise robustness in the steering case (dashed). 
Specifically, the efficiency thresholds (the intersection of the curves with the abscissa, to a good approximation) of maximally entangled states and the optimal state coincide. By contrast, in Bell-nonlocality, the efficiency thresholds of maximally entangled and optimal states differ greatly. Thus, the anomaly identified by Eberhard for Bell-nonlocality, between efficiency thresholds and entanglement, entirely vanishes for the simplest loophole-free steering. Indeed, once the detection efficiency exceeds the threshold, \emph{all} pure entangled states suffice to demonstrate quantum steering; see Theorem~3 in Section \ref{sm:1/X_bound} of the SM.
\blk

\emph{Discussions and conclusion.---} In this work, we have theoretically and experimentally found the minimal requirements for demonstrating quantum steering in the presence of detection inefficiencies.
Our protocol is, by far, the simplest to show detection-loophole-free steering, 
both in formal terms of minimum complexity cost and also in terms of having minimal experimental requirements, of only a single moderately high 
efficiency photon detector. 

By deriving a family of steering inequalities, analogous to Eberhard's Bell inequality, we also determine that the detection efficiency threshold for the simplest loophole-free quantum steering experiment is $50\%$, and we apply it to demonstrate EPR-steering with an efficiency of only $51.6\%$.
We discover---rather surprisingly---that for pure states this efficiency threshold is independent of how much entanglement the state has (as long as it is nonzero), unlike Bell-nonlocality, where the minimum efficiency bound to close the detection loophole is only available for almost unentangled states.
Moreover, we show that using a single detector on the untrusted side, of efficiency above $50\%$, incurs no penalty in terms of complexity compared to the idealized experiment with two $100\%$-efficient detectors.

Further investigations are necessary to determine if \textit{one-click} measurements are optimal in the case of a 
general steering scenario, where more than two stations and high-dimensional quantum states are involved \cite{srivastav2022quick,PhysRevLett.115.010402,PhysRevLett.111.250403, de2023complete,PhysRevA.84.032115,hao2022demonstrating,PhysRevLett.131.110201,PhysRevLett.132.210202,d2025semidefinite,pepper2024scalable}. 
In particular, an interesting direction is to study minimal requirements and noise tolerance in complex quantum networks involving independent sources \cite{jones2021network,sarkar2024network,li2025detecting,li2025measuring}.
Quantum steering is the foundational concept behind secure technologies like one-sided device-independent quantum communication and quantum key distribution protocols \cite{branciard2012one, slussarenko2022quantum,PhysRevLett.97.120405,zeng2025steering,comandar2016quantum} 
and certified randomness generation \cite{joch2022certified,li2024randomness}. 
Our results establish an operational benchmark for the resources required to faithfully implement these technologies, in terms of experimental complexity, and the requirements on the devices involved.

\emph{Acknowledgements.---}
This work was supported by the Australian Research Council Centre of Excellence Grant No.~CE170100012;  N.T.~is a recipient of an Australian Research Council Discovery Early Career Researcher Award (DE220101082). 
N.T.R.~and M.H.~acknowledge support by the Australian Government Research Training Program (RTP). M.H. acknowledges useful discussions with Eric G. Cavalcanti. E.P.~is a recipient of an Australian Research Council Discovery Early Career Researcher Award (DE250100762). This material is based upon work supported by the Air Force Office of Scientific Research under Award No.~FA2386-23-1-4086 (held by N.T., H.M.W., and T.J.B.).

%


\clearpage
\pagebreak
\widetext
\begin{center}
\textbf{\large Supplemental Material for: Detection-loophole-free nonlocality in the simplest scenario}\\[1.2em]

\begin{small}

Nandana T Raveendranath$^{1}$,
Travis J. Baker$^{1,2}$, 
Emanuele Polino$^{1}$,
Marwan Haddara$^{1}$,
Lynden K. Shalm$^{3}$,
Varun B. Verma$^{3}$
Geoff J. Pryde$^{1}$,
Sergei Slussarenko$^{1}$,
Howard M. Wiseman$^{1}$,
and Nora Tischler$^{1}$

\vspace{0.3cm}

\textit{$^{1}$\,Quantum and Advanced Technologies Research Institute, \\ Griffith University, Yuggera Country, Brisbane, QLD 4111, Australia}\\[0.3em]

\textit{$^{2}$\,Nanyang Quantum Hub, School of Physical and Mathematical Sciences,\\
Nanyang Technological University, Singapore 637371}\\[0.3em]

\textit{$^{3}$\,National Institute of Standards and Technology,\\
325 Broadway, Boulder, Colorado 80305, USA}
\end{small}

\end{center}

\setcounter{equation}{0}
\setcounter{figure}{0}
\setcounter{table}{0}
\setcounter{page}{1}
\makeatletter
\renewcommand{\theequation}{S\arabic{equation}}
\renewcommand{\thefigure}{S\arabic{figure}}
\renewcommand{\thetable}{S\arabic{table}}

\setcounter{secnumdepth}{2}

\tableofcontents

\section{Fundamental limits for steering with one detector}
\label{sm:1/X_bound}

First, we determine the fundamental thresholds for steering with one detector.
These thresholds apply to any number of one-click measurements, $X$. 
To determine cutoff efficiencies for steerability, we assume the probability of a detection event is independent of $x$, well represented by an average $\epsilon$.

From the Schmidt decomposition, there always exist sets of orthonormal bases $\{\ket{\alpha_i}\}_i$ and $\{\ket{\beta_i}\}_i$, such that $\ket{\Psi_{\mathrm{AB}}} = \sum_{i=0}^{d-1} \sqrt{\lambda_i} \ket{\alpha_i} \otimes \ket{\beta_i}$ for $d \coloneqq \min\{d_\mathrm{A}, d_\mathrm{B}\}$ and $\sum_i | \lambda_i | = 1$.
Moreover, any $\ket{\Psi_{\mathrm{AB}}}$ with a reduced state for Bob $\rho_\mathrm{B} = \sum_i \lambda_i \ketbra{\beta_i}$ is related to the maximally entangled state $\ket{\Phi^+} \coloneqq d^{-1/2} \sum_{i} \ket{ii}$ by local operations on Bob as $\ket{\Psi_{\mathrm{AB}}} = ( I \otimes \sqrt{d \rho_\mathrm{B}} V_\mathrm{B})\ket{\Phi^+}$, where $V_\mathrm{B}$ is unitary.
Observe that we can take $V_B=I_B$ without loss of generality, since EPR-steerability is invariant under local unitaries.
Therefore, by steering with one detector, the ``click'' ($a = +$)  outcomes reported by Alice steer Bob's system to the unnormalized states
\begin{align}
\sigma_{+|x} &= \Tr_\mathrm{A} \left[ (E_{+|x} \otimes I_B) \ketbra{\Psi_{\mathrm{AB}}}  \right] \\
&= \epsilon \sqrt{\rho_B} \; \Pi^T_{+|x} \sqrt{\rho_B}\;,
\end{align}
where the transpose is taken with respect to the Schmidt basis of Bob.
Alice's ``no-click'' effect is $\sigma_{\emptyset|x} = \rho_B - \sigma_{+|x}$.
For all entangled $\ket{\Psi_{AB}}$, this is proportional to a mixed state, with purity strictly increasing in $\epsilon$.

Here, we are interested in fundamental constraints on $\epsilon$ for steering with one detector.
We can formalize this as finding the $\epsilon$ as the following optimization problem (see \cite{pepper2024scalable}):
\begin{equation}
\begin{aligned}
\quad & \text{max} & & \epsilon &\\
& \text{s.~t.~} & & \sum_\lambda D_\lambda(+|x)\sigma_\lambda =  \epsilon \sqrt{\rho_B} \; \Pi^T_{+|x} \sqrt{\rho_B} & \forall \  x,\\
& & & \sum_\lambda \sigma_\lambda = \rho_B & \\
& & & \sigma_\lambda \geq 0 & \forall \  \lambda\;.
\end{aligned}
\label{eq:LTSDP_one_click}
\end{equation}
Here, $\{D_\lambda(a|x)\}_\lambda$ is the set of deterministic probability distributions that assign outcome $a$ for each value of $x$.
Since one-click measurements have dichotomic outcomes, there are $2^X$ such distributions.
The solution is the maximal value of the one-click detector efficiency $\epsilon$ such that $\{ \sigma_{+|x} \}_{x}$ admits an LHS model.
The first constraint ensures that the local-hidden-state ensemble $\{p_\lambda, \rho_\lambda\}$, with $\sigma_\lambda = p_\lambda \rho_\lambda$, correctly reproduces Bob's steered states.
The matrix inequality constraints ensure each state in this ensemble is physical, \emph{i.e.}~$\sigma_\lambda/\Tr[\sigma_\lambda]$ is positive semidefinite.

We begin by proving a lower bound for solutions to the optimization problem in \eqref{eq:LTSDP_one_click}, when one-click measurements are made on one half of an arbitrary entangled pure state. 
\begin{lemma}[Cutoff efficiency for one-detector steering tests]
\label{lem:cutoff_eff}
Let $\{ \epsilon \; \Pi_{+|x}\}_{x}$ be a set of rank-one effects defining $X$ one-click measurements, and $\ket{\Psi_{AB}}$ be an entangled state in $\mathcal{H}^{d_\mathrm{A}} \otimes \mathcal{H}^{d_B}$.
A local-hidden-state decomposition of the one-click assemblage exists if and only if 
\begin{equation}
\epsilon \leq \left[\maxeig{\sum_x \Pi_{+|x}}\right]^{-1},
\label{eq:spectral_radius_bound}
\end{equation}
where $\maxeig{A}$ is the maximum eigenvalue of $A$.
\end{lemma}
\begin{proof}
Define $\Pi_{+|x} \coloneqq \ketbra{e_x}$.
First, observe that every $\sqrt{\rho_B} \; \Pi^T_{+|x} \sqrt{\rho_B}$ is rank one, with eigenvector $\ket{w_x} = \sqrt{\rho_B} \ket{e_x}/\sqrt{\bra{e_x} \rho_B \ket{e_x}}$.
Define the projector $N_{x} \coloneqq I_B - \ketbra{w_x}$.
If the first set of constraints in \eqref{eq:LTSDP_one_click} are satisfied, we require that, $\forall \ ~x$,
\begin{equation}
\sum_\lambda D_\lambda(+|x) N_x \sigma_\lambda N^\dagger_{x} = 0  \;.
\end{equation}
Since every $\sigma_\lambda \geq 0$, any non-zero $\sigma_\lambda$ appearing in each of these constraints must be proportional to $\ket{w_x}$ itself.
For every $x$, each term in this sum is positive semidefinite, so the only way their sum vanishes is if every nonzero $\sigma_\lambda$ satisfies $N_x \sigma_\lambda=0$, \emph{i.e.}~it is proportional to $\ketbra{w_x}$.
Moreover, since each $\Pi_{+|x}$ (and hence each $\ket{e_x}$) is distinct, $\sigma_\lambda$'s corresponding to strategies that announce $+$ for more than one value of $x$ must be zero.
This means that for every $x$, there is exactly one strategy $\lambda(x)$ that announces $+$ for that setting and null outcomes otherwise, which implies that
\begin{equation}
\sigma_{\lambda(x)} = \epsilon \sqrt{\rho_B} \; \Pi^T_{+|x} \sqrt{\rho_B}
\end{equation}
to reproduce the steered states for the $+$ outcome.
The only remaining strategy, $\tilde{\lambda}$, is the one that announces only the null result. 
From the last two constraints in \eref{eq:LTSDP_one_click}, the corresponding operator must satisfy
\begin{equation}
\sigma_{\tilde{\lambda}} = \rho_B - \epsilon \sum\limits_x \sqrt{\rho_B} \; \Pi^T_{+|x} \sqrt{\rho_B} \geq 0 \;,
\end{equation}
which is equivalent to 
\begin{equation}
\sqrt{\rho_B}\left( I_B - \epsilon \sum\limits_x \Pi^T_{+|x} \right)\sqrt{\rho_B} \geq 0.
\end{equation}
Now, $\sqrt{\rho_B}$ is positive and invertible on its support, so this is equivalent to $\epsilon \sum_x \Pi^T_{+|x} \leq I_B$.
Hence, Eq.~\eqref{eq:spectral_radius_bound}.
\end{proof}

Interestingly, the efficiency threshold for pure-state one-detector steering scenarios is a property \emph{only} of the relative orientation of the ``click'' effects implemented by Alice.
In order to prove ultimate efficiency bounds on one-detector steering scenarios, the following property on the spectra of sums of rank-one matrices is required.
This involves only X, and the Hilbert Schmidt inner product between each pair of projectors, $\langle \Pi_x, \Pi_{x'}\rangle \coloneqq \Tr[\Pi_x^\dagger \Pi_{x'}]$.
\begin{lemma}
\label{lem:eigenvalues_of_projector_sum}
Let $\{\Pi_x\}_{x\in X}$ be a set of rank-one projection operators onto a 2-dimensional subspace of $H^{d_\mathrm{A}}$.
The spectrum of $\sum_x \Pi_x$ contains only two non-zero eigenvalues.
These are
\begin{equation}
\lambda_{\pm} = \frac{1}{2} \left( X \pm \left[ 2\left( \sum\limits_{x, x'} \langle \Pi_x, \Pi_{x'}\rangle \right) - X^2 \right]^{1/2} \right).
\label{eq:eigenvalues_of_projector_sum}
\end{equation}
\end{lemma}
\begin{proof}
Consider the eigendecomposition $\sum_x \Pi_x = UDU^\dagger$, where $U$ is unitary and $D$ is a diagonal matrix with real entries.
Since all projectors are supported on the same two-dimensional subspace, there exists a set of $d_\mathrm{A}-2$ orthonormal vectors $\{ \ket{v_i} \}$ such that $\mel{v_i}{\Pi_x}{v_i}= 0 ~\forall \  x,i$. 
Therefore, $U\Lambda U^\dagger$ and $(\sum_x \Pi_x)^2 = U\Lambda^2U^\dagger$ have exactly two nonzero eigenvalues, so that $\Tr\sum_x \Pi_x = \lambda_1 + \lambda_2 = X$, and $\Tr\left( \sum_x \Pi_x\right)^2 = \lambda_1^2 + \lambda_2^2 = X + \sum_{x\neq x'} \langle \Pi_x, \Pi_{x'}\rangle$.
These two equations imply that eigenvalues must be the roots of the quadratic equation $-2\lambda_1^2 + 2\lambda_1 X + \sum_{x\neq x'} \langle \Pi_x, \Pi_{x'}\rangle - X^2 = 0$.
The solutions are \eref{eq:eigenvalues_of_projector_sum}.
\end{proof}

This result allows, by an appropriate construction, uncovering the following ultimate bound on the efficiencies required for one-detector steering.
\begin{theorem}[Most inefficient one-detector steering]
\label{th:most_inefficient}
Let $\ket{\Psi_{AB}}$ be any entangled state in $\mathcal{H}^{d_\mathrm{A}} \otimes \mathcal{H}^{d_B}$.
There exists a set of $X$ one-click measurements defined by their click effects $\{ E_{+_|x}\}_{x}$, such that the corresponding one-detector steering test will demonstrate steering whenever
\begin{equation}
\epsilon > \frac{1}{X}.
\label{eq:thm_X_bound}
\end{equation}
\end{theorem}
\begin{proof}
We give an explicit construction for a set of measurements that approaches the $1/X$ bound from above.
Let $\ket{\alpha_0}$, $\ket{\alpha_1}$ be two vectors for Alice appearing in the Schmidt decomposition with non-zero coefficients.
From these, we construct normalized vectors that are real superpositions $\ket{e^\theta_x} \coloneqq \cos(\mu_x/2) \ket{\alpha_0} + \sin(\mu_x/2) \ket{\alpha_1}$, for $x=0,1,\dots,X-1$.
For some $0 < \theta < \pi/X$, we choose
\begin{equation}
\mu_x \coloneqq \left( x - \frac{X-1}{2} \right)\theta,
\end{equation}
so that the amplitudes are equally spaced by $\theta$, and distinct.
From these, we define the projectors $\Pi^\theta_x \coloneqq \ketbra{e^\theta_x}$, so that each ``click'' effect in each of Alice's POVMs is $E_{+_|x}=\epsilon \Pi^\theta_x$.
Therefore, 
\begin{align}
\sum_{x,x'} \langle \Pi^\theta_x, \Pi^\theta_{x'}\rangle
&= \sum_{x,x'} \cos^2 \left( \frac{(x-x')\theta}{2} \right) \\
&= \frac{1}{2} \left( X^2 + \sum_{j=-(X-1)}^{X-1} (X-|j|)\cos(j\theta) \right) \\
&= \frac{1}{2} \left( X^2 + \left(\frac{\sin(X\theta/2)}{\sin(\theta/2)}\right)^2 \right).
\end{align}
Now, Lemma~\ref{lem:cutoff_eff} implies that steering is possible for this construction when
\begin{align}
\epsilon &> \inf\limits_\theta \left[\maxeig{\sum_x \Pi^\theta_{x}}\right]^{-1},
\label{eq:spectral_radius_theta_family}
\end{align}
and by Lemma~\ref{lem:eigenvalues_of_projector_sum} we evaluate the largest eigenvalue as
\begin{equation}
\maxeig{\sum_x \Pi^\theta_{x}} = \frac{1}{2} \left(X +  \frac{\sin(X\theta/2)}{\sin (\theta/2)}\right).
\end{equation}
The infimum in Eq.~\eqref{eq:spectral_radius_theta_family} is approached in the limit $\theta\rightarrow 0$, proving Eq.~\eqref{eq:thm_X_bound}.
\end{proof}

As mentioned in the main text, taking $X=2$ provides the simplest loophole-free demonstrations of steering.
This permits steering with $W=12$, \emph{and} can be witnessed with one detector on the untrusted party if the sole efficiency surpasses 1/2.

\section{Optimal witnesses for simplest loophole-free steering}
\label{sm:witness_derivation}

For the simplest steering test with $X=2$, here we show that for the measurements used in the construction to prove Theorem~\ref{th:most_inefficient}, we can derive exact steering inequalities in closed form.

We consider a bipartite pure state in Schmidt form $\ket{\Psi_{\mathrm{AB}}} = \sum_{i=0}^{d-1} \sqrt{\lambda_i} \ket{\alpha_i} \otimes \ket{\beta_i}$ for $d \coloneqq \min\{d_\mathrm{A}, d_\mathrm{B}\}$ and $\sum_i | \lambda_i | = 1$.
Consider relabelling of the Schmidt decompositions into an ordering $\lambda_0 \geq \lambda_1 \geq \cdots \geq \lambda_{d-1}$, and assume $\ket{\Psi_{\mathrm{AB}}}$ is entangled, guaranteeing $\lambda_1>0$.
We take Alice's two one-click effects to be $E_{+_|x}=\epsilon \; \Pi^\theta_x = \epsilon \ketbra{e^\theta_x}$ for $x=0,1$, with 
\begin{align}
\ket{e^\theta_x} = \cos(\theta/2) \ket{\alpha_0} + (-1)^x \sin(\theta/2) \ket{\alpha_1}.
\end{align}
For the click outcome $+$, Bob's (unnormalized) steered states are
\begin{align}
\sigma^\theta_{+|x} &= \epsilon \Tr_A \left[ (\Pi^\theta_{+|x} \otimes I_B)\ketbra{\Psi_{AB}} \right] \\
&= \epsilon\left( \lambda_0\cos^2\frac{\theta}{2} \ketbra{\beta_0} 
+ \sqrt{\lambda_0 \lambda_1} \cos\frac{\theta}{2}\sin\frac{\theta}{2} \hat{V}
+ \lambda_1\sin^2\frac{\theta}{2} \ketbra{\beta_1}\right).
\end{align}
Here, $\hat{V} \coloneqq \dyad{\beta_0}{\beta_1} + \dyad{\beta_1}{\beta_0}$ is the flip operator between the two largest eigenvectors of $\rho_B$.

\subsection{The primal problem: an ansatz}
\label{sm:sec_primal_ansatz}

To certify EPR-steerability, we use the feasibility problem for testing if an input assemblage $\{ \sigma_{a|x} \}_{a,x}$ admits an LHS model from \cite{Cav16}.
This is:
\begin{equation}
\begin{aligned}
\quad & \text{max} & & \mu &\\
& \text{s.~t.~} & & \sum_\lambda D(a|x,\lambda)\sigma_\lambda = \sigma_{a|x} & \forall \  a,~x,\\
& & & \sigma_\lambda \geq \mu I_B & \forall \  \lambda\;.
\end{aligned}
\label{eq:SM_eig_SDP}
\end{equation}
Here, the primal variables are the real scalar $\mu$ and the $A^X$ positive semi-definite matrices $\{ \sigma_\lambda \}_\lambda$.
The former places a lower bound on the spectrum of the latter, meaning that a non-negative result implies the existence of an LHS model for the input assemblage.
The four equality constraints read
\begin{align}
\sigma_{0} + \sigma_{1} &= \sigma_{+|0} \label{eq:SM_primal_equalities_first}\\
\sigma_{0} + \sigma_{2} &= \sigma_{+|1} \\
\sigma_{2} + \sigma_{3} &= \rho_B - \sigma_{+|0} \\
\sigma_{1} + \sigma_{3} &= \rho_B - \sigma_{+|1}\;. \label{eq:SM_primal_equalities_last}
\end{align}
Based on numerical results, we make an ansatz for the structure of the primal variables of the form:
\begin{align}
\sigma_{0} &= a_0 \ketbra{\beta_0} + c_0 \ketbra{\beta_1} \\
\sigma_{1} &= a_1 \ketbra{\beta_0} + b_1 \hat{V} + c_1 \ketbra{\beta_1} \\
\sigma_{2} &= a_1 \ketbra{\beta_0} - b_1 \hat{V} + c_1 \ketbra{\beta_1} \\
\sigma_{3} &= c_0 \ketbra{\beta_0} + c_3 \ketbra{\beta_1}
\end{align}
for real coefficients $\{a_0, a_1, b_1, c_0, c_1, c_3\}$.
Also, we make the ansatz that the least eigenvalues of these operators are equal, $\mineig{\sigma_\lambda}=0~\forall \ \lambda$, so that equality holds for matrix inequality in \eqref{eq:SM_eig_SDP}.
Under these conditions, using Eqs.~\eqref{eq:SM_primal_equalities_first}--\eqref{eq:SM_primal_equalities_last} we must have
\begin{align}
a_0+a_1 &= \expval{\sigma_{+|0}}{\beta_0} \\ 
b_1 &= \frac{1}{2} \Tr[\hat{V}\sigma_{+|0}]\\
c_0+c_1 &= \expval{\sigma_{+|0}}{\beta_1} \\ 
a_1+a_3 &= \lambda_0 - \expval{\sigma_{+|0}}{\beta_0} \\
c_1+c_3 &= \lambda_1 - \expval{\sigma_{+|0}}{\beta_1} \\ 
\end{align}
This gives five equations for six unknowns, which we can solve algebraically up to one remaining degree of freedom; we refer the reader to the Mathematica notebook at \cite{code_repo}.
The maximum of the objective $\mu$ attainable in \eqref{eq:SM_eig_SDP} is simply $\mu^\star = \min_\lambda \mineig{\sigma_\lambda} = c_0$.
\newcommand\ssz{\tilde{\sigma}_{z}}
Defining the operator $\ssz \coloneqq \dyad{\beta_0}{\beta_0} - \dyad{\beta_1}{\beta_1}$, this value can be expressed succinctly as
\begin{equation}
\mu^\star = \frac{1}{4} \left(\lambda_0 - \Tr[\ssz\sigma_{+|x}] - \sqrt{\Tr[\hat{V}(\sigma_{+|0} - \sigma_{+|1})]^2 + K_+^2}\right),
\label{eq:primal_solution}
\end{equation}
where $K_+ \coloneqq \lambda_0 - \Tr[\sigma_{+|x}] \geq 0$.
At this point, we have not proven optimality of the solution---only that the ansatz above is a feasible set of variables---\emph{i.e.}~that it satisfies the constraints of problem \eqref{eq:SM_eig_SDP}.
In the next subsection, we prove optimality by finding a solution to the dual optimization problem that attains the same value in Eq.~\eqref{eq:primal_solution}.

\subsection{Solving the dual problem with zero duality gap}
\label{sm:sec_dual_ansatz}

The dual program to \eqref{eq:SM_eig_SDP} is
\begin{equation}
\begin{aligned}
\quad & \text{min} & & \sum\limits_{a,x} \Tr F_{a|x} \sigma_{a|x} &\\
& \text{s.~t.~} & & \sum\limits_{a,x} D(a|x,\lambda)F_{a|x} \geq 0 & \forall \  \lambda,\\
& & & \sum\limits_{a,x,\lambda} \Tr D(a|x,\lambda)F_{a|x} = 1\;. &
\end{aligned}
\label{eq:SM_SE_SDP}
\end{equation}
The dual variables are Hermitian operators $\{ F_{a|x} \}_{a,x}$.
A negative value of the objective function $\sum_{a,x} \Tr F_{a|x} \sigma_{a|x}$ certifies that no LHS decomposition exists for that assemblage $\{ \sigma_{a|x} \}_{a,x}.$

We are required to construct a set of dual variables for the problem \eqref{eq:SM_SE_SDP} that satisfy the constraints and attain the value of the primal derived above, so that $\sum_{a,x} \Tr F^\star_{a|x} \sigma_{a|x} = \mu^\star$.
We make an ansatz for the four $F_{a|x}$ operators in terms of 
\begin{align}
F_{+|0} &= a_{+|0} \dyad{\lambda_0} + b_{+|0} \hat{V} + c_{+|0} \dyad{\lambda_1} \\
F_{+|1} &=-a_{+|0} \dyad{\lambda_0} - b_{+|0} \hat{V} + c_{+|0} \dyad{\lambda_1}  \\
F_{\emptyset | 0} &= a_\emptyset \dyad{\lambda_0} \\
F_{\emptyset | 1} &= 0\;,
\end{align}
for some real numbers $a_{+|0},\,  b_{+|0},\, c_{+|0},\, a_\emptyset$.
The equality constraint in \eqref{eq:SM_SE_SDP} requires
\begin{equation}
a_\emptyset = \frac{1}{2}(1 - 4c_{+|0}).
\end{equation}
Now, the matrix inequality constraints in \eqref{eq:SM_SE_SDP} are
\begin{align}
F_{+|0} + F_{+|1} &\geq 0 \label{eq:SM_dual_mi1}\\
F_{+|0} + F_{\emptyset|1} &\geq 0 \label{eq:SM_dual_mi2}\\
F_{\emptyset|0} + F_{+|1} &\geq 0 \label{eq:SM_dual_mi3}\\
F_{\emptyset|0} + F_{\emptyset|1} &\geq 0\;. \label{eq:SM_dual_mi4}
\end{align}
To solve for the degrees of freedom in the ansatz, we impose that equality holds in each of these inequality constraints, \emph{i.e.} the left side of each constraint has a vanishing least eigenvalue.
This will be true for ~\eqref{eq:SM_dual_mi1} and \eqref{eq:SM_dual_mi4} if $c_{+|0}\geq 0$ and $a_\emptyset \geq 0$, while the remaining two conditions imply
\begin{align}
b_{+|0}^2 &= a_{+|0} c_{+|0} \\
a_\emptyset &= 2 a_{+|0} \\
c_{+|0} &= \frac{1}{4} - 2a_{+|0} \;.
\end{align}
These lead to a natural trigonometric parametrization in terms of a single real parameter $\gamma$.
Defining $a_{+|0} \coloneqq \frac{1}{4} \sin^2(\gamma/4)$, the double angle formulas allow the coefficients to be expressed as
\begin{align}
a_{+|0} &= \frac{1}{8} \left(1-\cos\frac{\gamma}{2}\right) \\
b_{+|0} &= \frac{1}{8} \sin\frac{\gamma}{2}\\
a_\emptyset &= \frac{1}{4} \left(1-\cos\frac{\gamma}{2}\right) \\
c_{+|0} &= \frac{1}{8} \left(1+\cos\frac{\gamma}{2}\right) \;.
\end{align}
It remains to minimize the achievable value of the dual objective function in terms of $\gamma$, which amounts to finding solutions to
\begin{equation}
\frac{\partial}{\partial \gamma} \sum\limits_{a,x} \Tr F^\gamma_{a|x} \sigma_{a|x} = 0.
\end{equation}
Straightforward computation detailed in the Mathematica notebook at \cite{code_repo} shows that this occurs for 
\begin{equation}
\gamma = 2 \tan^{-1} \left( \frac{\sum_x \Tr[\hat{V} (\sigma_{+|0} - \sigma_{+|1})]}{K_+} \right),
\label{eq:SM_gamma_argmin}
\end{equation}
where we have again defined $K_+ \coloneqq \lambda_0 - \Tr[\sigma_{+|x}] \geq 0$.
Therefore, the value attained by the steering functional is
\begin{equation}
\sum\limits_{a,x} \Tr F^\star_{a|x} \sigma_{a|x} = \frac{1}{4} \left(\lambda_0 - \Tr[\ssz\sigma_{+|x}] - \sqrt{\Tr[\hat{V}(\sigma_{+|0} - \sigma_{+|1})]^2 + K_+^2}\right).
\label{eq:SE_val_exact}
\end{equation}
This exactly matches the value obtained by the primal ansatz, Eq.~\eqref{eq:primal_solution}.
To summarize, the exact closed forms of the optimal witness $\{F^\star_{a|x}\}$
are:
\begin{align}
F_{+|0} &= \frac{1}{8} \left[ \left(1-\cos\frac{\gamma}{2} \right)\dyad{\lambda_0} + \sin\frac{\gamma}{2} \hat{V} + \left(1+\cos\frac{\gamma}{2} \right)\dyad{\lambda_1} \right] \\
F_{+|1} &= \frac{1}{8} \left[ \left(\cos\frac{\gamma}{2} - 1\right)\dyad{\lambda_0} - \sin\frac{\gamma}{2} \hat{V} + \left(1+\cos\frac{\gamma}{2} \right)\dyad{\lambda_1} \right] \\
F_{\emptyset | 0} &= \frac{1}{4} \left(1-\cos\frac{\gamma}{2} \right) \dyad{\lambda_0} \\
F_{\emptyset | 1} &= 0\;,
\end{align}
where $\gamma(\epsilon,\theta, \lambda_0, \lambda_1)$ is defined through Eq.~\eqref{eq:SM_gamma_argmin}.

\section{Simplest Bell-nonlocality: Eberhard's inequality}
\label{sec:eberhard_overview}

The lowest bound for the quantum critical detector efficiency for Bell-nonlocality demonstrations $\epsilon_{B}$  was derived by Eberhard \cite{Eberhard1993} to be $\epsilon_B = 2/3$. He showed, using a limiting process, that in the absence of background, there are two-qubit entangled states that can be used to demonstrate nonlocality with detector efficiencies arbitrarily close to the critical efficiency. This bound is notably larger than the steering bound $\epsilon = 1/2$ and, furthermore, in contrast to the steering case,  may only be approached with non-maximally entangled states as illustrated in Fig.~3 in the main text. In this Section, we show how Eberhard's 2/3 efficiency bound may be derived, and describe how the plots for the Bell-nonlocality noise-efficiency trade-off thresholds in Fig.~3 in the main text are obtained. 

\subsection{The Eberhard inequality}
The simplest Bell scenario has two parties, $A$ and $B$,  each with two inputs $x,y \in \{0,1\}$ and two outputs $a_x,b_y \in \{0,1\}$ for every input \cite{Saunders2012}. The object of interest for non-locality demonstrations is the set of joint probabilities $\wp(ab|xy)$ which are connected to the observed statistics via $\wp(ab|xy) \simeq {N_{ab|xy}}/{N_{xy}}$. Here, $N_{ab|xy}$ is the number of counts of type $ab$ for setting choices $xy$  and $N_{xy}$ is the total number of observed counts, \emph{i.e.}, $N_{xy} = \sum_{ab}N_{ab|xy}$, for the settings $xy$.  To allow for inefficient detectors, one may add a third `null' outcome $\emptyset$ for every input, corresponding to the non-detection event.  The experiment thus samples a behaviour $\tilde \wp(ab|xy)$ with three outputs $ a,b \in \{0,1, \emptyset\}$ for every input,  $x,y \in \{0,1\}$. The assumption that this object is consistent with a locally causal model is equivalent to demanding that 
\begin{align}
    \tilde \wp(ab|xy) = \int_{\Lambda} D(a|x,\lambda) D(b|x,\lambda)p(\lambda)d\lambda = \int_{\Lambda_{a|x}\cap \Lambda_{b|y} }p(\lambda)d\lambda = \mu[\Lambda_{a|x}\cap \Lambda_{b|y}]
\end{align}
for all $a,b , x, y$. The LHV space $\Lambda$  may thus be understood to be split to regions $\Lambda_{a|x}$ and $\Lambda_{b|y}$ which specify the probabilities for events corresponding to measures of appropriate subsets $U\subset \Lambda$ via the probability measure $\mu[U] = \int_{U}p(\lambda)d\lambda$. Using the properties of the measure $\mu(\bullet)$, namely that $\mu(\Lambda) = 1,\;  \mu(\cup_i U_i) = \sum_i \mu(U_i)$ for mutually disjoint sets $U_i\subset \Lambda$ and $\mu(U_1) \leq \mu(U_2)$ whenever $U_1 \subset U_2$, we may rephrase Eberhard's argument \cite{Eberhard1993} following the kind of technique used in Ref.~\cite{Ghirardi2008}. 

Note that the sets $\Lambda_{a|x}$ and $\Lambda_{b|y}$ form a disjoint partition of $\Lambda$ for each $x,y$ so that, for example,  $\cup_{a}\Lambda_{a|x} = \Lambda$ and $\Lambda_{a|x}\cap \Lambda_{a'|x} \neq \emptyset$ iff $a=a'$ holds for all $x$. Consider then the measure of the set $ \mathcal{A}=(\Lambda_{a=0|x=0} \cap \Lambda_{b=0|y=0} \setminus  \Lambda_{a\neq0|x=1}) \setminus\Lambda_{b\neq 0 |y=1}$. Since $U\setminus U' = U\cap U'^{c}$, where the superscript $c$ indicates the complementary set, $(\Lambda_{a\neq0 |x=1})^c = \Lambda_{a=0|x=1}$ and similarly $(\Lambda_{b\neq 0|y=1})^c = \Lambda_{b=0|y=1}$,  the measure of this set intuitively represents the likelihood of the variable $\lambda$ to force the outcomes $a_x,b_y=0$ for all the measurements $xy \in \{00,01,10\}$.  The following relation holds:  

\begin{align}
     \mu(\mathcal{A}) \leq \mu(\Lambda_{a=0|x=1}\cap \Lambda_{b=0|y=1}), \label{Equation:SupplementalFromAboveEq.}
\end{align}
as $(U_1\setminus U_2)\setminus U_3 = (U_1\cap U_2^c\cap U_3^c) \subset (U_2^c\cap U_3^c)$. On the other hand, since for any $U, U'$ it holds that $\mu(U\setminus U') = \mu(U) - \mu(U\cap U')$, one can show that:
\begin{align}
 \mu \left( (U_1\setminus U_2) \setminus U_3 \right) &= \mu(U_1 \setminus U_2) - \mu((U_1\setminus U_2) \cap U_3) = \mu(U_1) - \mu(U_1\cap U_2) - \mu(U_1 \cap U_2^c \cap U_3)\\
 & \geq \mu(U_1) - \mu(U_1\cap U_2) - \mu(U_1\cap U_3),
\end{align}
where, to get to the inequality on the second line, the positivity of $\mu(\bullet)$ has been used along with the fact that $(U_1\cap U_2^{c}\cap U_3 )\subset (U_1\cap U_3)$. By setting $U_1 = \Lambda_{a=0|x=0}\cap \Lambda_{b=0|y=0}, U_2 = \Lambda_{a\neq0|x=1}$ and $U_3 = \Lambda_{b\neq0|y=1}$, therefore 
\begin{align}
    \mu(\mathcal{A}) &\geq \mu( \Lambda_{a=0|x=0}\cap \Lambda_{b=0|y=0}) - \mu( (\Lambda_{a=0|x=0}\cap \Lambda_{b=0|y=0})\cap \Lambda_{a\neq0|x=1}) - \mu((\Lambda_{a=0|x=0}\cap \Lambda_{b=0|y=0})\cap \Lambda_{b\neq0|y=1}) \\
    & \geq \mu(\Lambda_{a=0|x=0}\cap \Lambda_{b=0|y=0}) -  \mu( \Lambda_{b=0|y=0}\cap \Lambda_{a\neq0|x=1}) - \mu(\Lambda_{a=0|x=0}\cap \Lambda_{b\neq0|y=1}) \\
    &= \mu(\Lambda_{a=0|x=0}\cap \Lambda_{b=0|y=0}) - \sum_{a' \in \{1,\emptyset \}}\mu(\Lambda_{b=0|x=0}\cap \Lambda_{a = a'|x=1}) -  \sum_{b'\in \{1,\emptyset\}}\mu(\Lambda_{a=0|x=0}\cap \Lambda_{b = b'|y=1}). \label{Equation:SupplementalFromBelowEq.}
\end{align}
Here, to get to the last line, the additivity of $\mu(\cdot)$ over disjoint subsets has been used. Combining the constraints on $\mu(\mathcal{A})$ from~\eqref{Equation:SupplementalFromAboveEq.} and \eqref{Equation:SupplementalFromBelowEq.} and writing them in terms of the probabilities $\tilde{\wp}(ab|xy)$, we get the relation
\begin{align}
   E= \tilde{\wp}(00|11) + \tilde{\wp}(10|10) + \tilde{\wp}(\emptyset0|10) + \tilde{\wp}(01|01) + \tilde{\wp}(0\emptyset|01)  - \tilde{\wp}(00|00) \geq 0 \label{Equation:Eberhard'sInequality}
\end{align}
as a necessary conditon for $\tilde{\wp}(ab|xy)$ to be compatible with a local hidden variable model.  Eq.~\eqref{Equation:Eberhard'sInequality} is Eberhard's inequality \cite{Eberhard1993} written in terms of the probabilities instead of the counts $N_{ab|xy}$. 

\subsection{Optimal noise-efficiency thresholds}
We now retrieve the optimal thresholds allowed by quantum states and measurements. 
Considering first ideal two-qubit measurements in the Bloch representation leads to projection-valued measures with effects of the form $\hat{P}_{a=0|x} = \frac{1}{2}(I +\vec{n}^x\cdot \vec{\sigma} ), \hat{P}_{a=1}=\frac{1}{2}(I-\vec{n}^x\cdot \vec{\sigma})$. Here,
$\vec{n}^x = (n_1^x, n_2^x, n_3^x)$ with $n_i^x \in[0,1]$ and $|\vec{n}^x|=1$ is the Bloch vector of the projection and $\vec{\sigma} = (\sigma_1, \sigma_2, \sigma_3)$ is a vector consisting of the Pauli matrices 
\begin{align}
\sigma_1 = \begin{pmatrix}
    0 & 1 \\
    1&0
\end{pmatrix},
\hspace{0.3cm}
\sigma_2 = \begin{pmatrix}
    0& -i\\
    i& 0
\end{pmatrix}, \hspace{0,3cm} \sigma_3 = \begin{pmatrix}
    1 & 0 \\
    0 & -1
\end{pmatrix}.
\end{align}Similar representation may be taken for Bob with projectors $\hat{P}_{b|y}$, $b_y \in \{0,1\}$ so that probabilities are computed via $\wp(ab|xy) = \Tr[\hat{P}_{a|x} \otimes \hat{P}_{b|y}\rho]$. When imperfect detector efficiency $\epsilon$ is taken into account, the projectors for $a,b \in \{0,1\}$ are replaced by effects $\hat{E}_{a|x} = \epsilon \hat{P}_{a|x}$ and $ \hat{E}_{b|y} = \epsilon \hat{P}_{b|y}$ and a third effect $\hat{E}_{\emptyset|x} = I - \sum_{a\in \{0,1\}}\hat{E}_{a|x} = (1- \epsilon)I$, $\hat{E}_{\emptyset|y} = (1-\epsilon)I$ is added for every input $x,y$ corresponding to the inclusion of null outcome. The probabilities are now computed from $\tilde{\wp}(ab|xy) = \Tr[\hat{E}_{a|x}\otimes \hat{E}_{b|y}\rho]$ where $a,b \in \{0,1, \emptyset\}$ for all $x,y$. 

Let $\hat{B}$ denote a quantum Bell operator, using which inequality \eqref{Equation:Eberhard'sInequality} is expressed as $\langle \hat{B}\rangle = \Tr[\hat{B}\rho] \geq 0$. This inequality can be violated if $\hat{B}$ has at least one negative eigenvalue. Since the eigenvalues are invariant under unitary operations, it is possible to choose measurements with $\vec{n}^{x=1}= \vec{n}^{y=1} = (1,0,0)$ and $\vec{n}^{x=2}=(n_1^x, n_2^x, 0)$, $\vec{n}^{y=2}=(n_1^y, n_2^y, 0)$  without loss of generality. Furthermore, since $ |\vec{n}^x| = |\vec{n}^y| =1$, it is possible to introduce parameters $\phi^x, \phi^y$ so that $n^x_{1} = \mathrm{Re}\left[ e^{- i\phi^x }  \right], n^x_{2} = \mathrm{Im}\left[e^{-i\phi^x} \right]$ and similarly for $n^y_{1/2}$. Now for example, $\hat{E}_{a=0|x=1} = \epsilon/2(I + n^x_1 \sigma_1 + n_2^x \sigma_2)$ is represented by the matrix $\hat{E}_{a=0|x=1} = \epsilon/2 
\left(\begin{smallmatrix}
    1 & e^{i\phi^x} \\
    e^{-i\phi^x} &1
\end{smallmatrix}\right)$. Using the convention of \cite{Eberhard1993} and defining $T = \epsilon/2(e^{i\phi^x}-1)$ and $R = (e^{i\phi^y}-1)$, it is found that the Bell operator can be expressed as 
\begin{align}
    \hat{B} = \epsilon/2\begin{pmatrix} (2- \epsilon )& (1- \epsilon) & (1- \epsilon) &TR- \epsilon \\
    (1- \epsilon ) & (2- \epsilon ) & TR^* - \epsilon & (1- \epsilon) \\
    (1- \epsilon) & T^*R - \epsilon & (2- \epsilon) & (1- \epsilon) \\
    T^*R^* - \epsilon & (1- \epsilon) & (1- \epsilon) &(2- \epsilon) \label{Equation:EberhardsBellOperator}
    \end{pmatrix}\;.
\end{align}
Equation \eqref{Equation:EberhardsBellOperator} is exactly the matrix derived by Eberhard \cite{Eberhard1993}. Following Eberhard, the critical value  $\epsilon = 2/3$ may be  verified by using the fact that the determinant $\mathrm{det}[\hat{B}]$ of the matrix turns from negative to positive when the \emph{last} negative eigenvalue changes sign from negative to positive. This condition can be checked computationally  by sweeping over the angles $\phi^x, \phi^y$. We will show in Sec.~\ref{SupplementalSection:BlackBoxEberhard} of the SM how this bound can be obtained analytically, and that it applies to general no-signalling theories as well. For now, however,  we focus on obtaining the full noise-efficiency trade-off for the quantum model above. 

The effect of white noise in the quantum state can be baked into the matrix  representation of the Bell operator in Eq.~\eqref{Equation:EberhardsBellOperator} by redefining it as $\hat{B}_{\eta} = (1-\eta)\hat{B} + \eta/4  \Tr[\hat{B}]\times I_{\mathcal{H}_{AB}}$, with $\mathcal{H}_{AB} = \mathbb
 C^2 \otimes \mathbb{C}^2$. This follows from the fact that, for a mixed state $\rho$ of the form  $\rho = (1-\eta)\ket{\psi_{AB}}\bra{\psi_{AB}} + \eta \dfrac{1}{4}I_{\mathcal{H}_{AB}}$ the expectation value of the Bell operator $\hat{B}$ becomes 
 \begin{align}
     \Tr[\hat{B}\rho] &= (1-\eta)\Tr[\hat{B} \ket{\psi_{AB}}\bra{\psi_{AB}}] + \dfrac{\eta}{4}\Tr[\hat{B}] = \Tr\left[\hat{B}_{\eta}\ket{\psi_{AB}}\bra{\psi_{AB}}\right]\;. \label{equation:NoisyBelloperatoreEquality}
 \end{align}
Therefore, finding the necessary condition for violation of Eq.~\eqref{Equation:Eberhard'sInequality} for a white-noise mixed state can be re-stated as the operator $\hat{B}_{\eta}$ having at least one negative eigenvalue. Note that $\hat{B}$ and $\hat{B}_{\eta}$ are diagonal in the same basis, and so share the same eigenvectors. Let $\ket{\xi_{\min}}$ be the eigenvector of $\hat{B}_{\eta}$ with the \emph{least} eigenvalue $\mineig{\hat{B}_{\eta}}$, or equivalently, the  eigenvector of $\hat{B}$ with the least eigenvalue $\mineig{\hat{B}}$.  The least eigenvalue $\mineig{\hat{B}_{\eta}}$ of $\hat{B}_{\eta}$ turns zero when  
\begin{align}
    \Tr \left[\hat{B}_{\eta} \ket{\xi_{\min}}\bra{\xi_{\min}} \right]&= (1-\eta) \Tr\left[\hat{B} \ket{\xi_{\min}}\bra{\xi_{\min}} \right]+ \dfrac{\eta}{4}\Tr[\hat{B}]= (1-\eta)\mineig{{\hat{B}}} + \dfrac{\eta \epsilon}{2}(2-\epsilon)  =0 \label{Equation:leastEigenvalueEquation}
\end{align}
so that
\begin{align}
     \eta = \dfrac{-\mineig{\hat{B}}}{\epsilon/2(2-\epsilon) -\mineig{\hat{B}}} = \dfrac{1}{1- \left(\dfrac{\epsilon(2-\epsilon)}{2\mineig{\hat{B}}}\right)}. \label{Equation:Noise-Efficiencyequality}
\end{align}
The value of $\mineig{\hat{B}}$ where this equality holds can be solved from $\hat{B}$  computationally for a given value of $\epsilon$, by minimizing the determinant in terms of the angles $\phi^x, \phi^y$ and computing the least eigenvalue. We provide Python code which performs this computation in \cite{code_repo}. This code gives the black `optimized' curve in Fig.~3 in the main text for optimal trade-off between white noise and detector efficiency for nonlocality violations. 

\subsection{Noise-efficiency thresholds for maximally entangled states}

A different trade-off graph is obtained if the quantum state is fixed to be maximally entangled. With the chosen representation where Bloch vectors of the inputs were of the form $\vec{n}^{x=0} = \vec{n}^{y=0} = (1,0,0)$ and $\vec{n}^{x/y=1} = (n_1^{x/y},n^{x/y}_2,0)$ some states are better choices than others. From Eq.~\eqref{equation:NoisyBelloperatoreEquality} it is seen that the optimal maximally entangled state is generally that which minimizes $\Tr[\hat{B}\ket{\psi^{\max}}\bra{\psi^{\max}}]$. However, for the case of perfect detectors with $\epsilon=1$ there is no need to introduce a null outcome, and the optimal choice for the state is simply that which reaches the quantum bound for the violation of a Bell inequality, such as the Clauser-Horne (CH)-inequality \cite{Clauser1974} 
\begin{align}
  S =   \wp(00|11) - \wp(00|10) - \wp(00|01) -\wp(00|00) +  \wp (a=0|x=0)  +   \wp(b=0|y=0)\geq 0. \label{eQUATION:CH-inequality}
\end{align}
The optimal quantum violation of Eq.~\eqref{eQUATION:CH-inequality} is achieved by the value $\frac{1}{2}(1-\sqrt{2}) \simeq -0.207$. Sets of projective measurements and maximally entangled states are known which reach this value (see e.g. Ref.~\cite{Eberhard1993}), and such choices can be mapped to the chosen parametrization of the form of the inputs.   A state that reaches the maximum value for this parametrization is  $\ket{\Psi^+_{\pi/4}} = \dfrac{1}{\sqrt{2}}( \ket{HV} + e^{i \pi/4} \ket{VH})$ with Bloch vectors $\vec{n}^{x=0}=\vec{n}^{y=0}=(1,0,0)$ , $\vec{n}^{x=1}=(0,1,0)$ and $\vec{n}^{y=1}=(0,-1,0)$. The same state and measurement angles remain optimal even in the presence of detector inefficiency and with the addition of white noise. To see this, note that the no-signalling constraints $\sum_{a\in \{\emptyset, 0,1 \}}\tilde{\wp}(a0|x 1)= \tilde{\wp}(b=0|y=1)$ and $\sum_{b\in \{\emptyset, 0, 1\}} \tilde{\wp}(0b|1y) = \tilde{\wp}(a=0|x=1)$ imply that the Eberhard inequality \eqref{Equation:Eberhard'sInequality} may be equivalently written as 

\begin{align}
  E=  \tilde{\wp}(00|11) - \tilde{\wp}(00|10) - \tilde{\wp}(00|01) -\tilde{\wp}(00|00) +  \tilde{\wp}(a=0|x=0)  +   \tilde{\wp}(b=0|y=0)\geq 0
\end{align}
and so, by taking into account the form of the measurements, it follows that 
\begin{align}
    E &= \epsilon^2\left[\wp(00|11) - \wp(00|10) - \wp(00|01) -\wp(00|00)\right] +\epsilon [\wp (a=0|x=0)  +   \wp(b=0|y=0)] \\
    &= \epsilon^2 \left[ S - \wp (a=0|x=0)  -  \wp(b=0|y=0) \right] + \epsilon [\wp (a=0|x=0)  +   \wp(b=0|y=0)] \geq 0.
\end{align}
Note that here $S$ and the probabilities denoted by $\wp$ refer to the CH-expression of Eq.~\eqref{eQUATION:CH-inequality} evaluated with the ideal detector model. 
For a state $\rho = (1-\eta)\ket{\psi^{\max}}\bra{\psi^{\max}} + \eta \dfrac{1}{4}I_{\mathcal{H}_{AB}} $ (and rank 1 effects) the local marginals $\wp (a=0|x=0), \, \wp(b=0|y=0) $  equal $1/2$. Furthermore, the expression $S$ decomposes into $(\eta S^{\max}) + 1/2(1-\eta)$, where $\eta S^{\max}$  is the contribution from the maximally entangled component and $1/2(1-\eta)$ from the maximally mixed component of the state $\rho$. Thus, the relation
 \begin{align}
     \epsilon[(\eta S^{\max}+1/2 (1-\eta)-1)] \geq -1 \label{equation:noise-efficiencyforMAXENT}
 \end{align}
 is obtained for every maximally entangled state mixed with white noise. The smallest $\epsilon$ for a given $\eta$ for which a violation is possible is obtained when the expression inside the brackets is minimized, or equivalently, when $S^{\max}$ obtains its smallest value i.e. $\frac{1}{2}(1-\sqrt{2})$, which can be reached in the chosen parametrization with the state $\ket{\Psi^+_{\pi/4}} = \dfrac{1}{\sqrt{2}}( \ket{HV} + e^{i \pi/4} \ket{VH})$ as stated before. Thus the trade-off encoded in Eq.~\eqref{equation:noise-efficiencyforMAXENT} with value $S^{\max} = \frac{1}{2}(1-\sqrt{2})$ is optimal. Plotting the equality gives the green curve in Fig.~3 of the main text.

\subsection{Efficiency threshold for the simplest scenario\label{section:SM,EfficiencyIntheSimplestScenario}}
The Eberhard inequality \eqref{Equation:Eberhard'sInequality}  
deals with the detector imperfections by introducing a third outcome corresponding to the null event. This arguably increases the complexity $W = A^XB^Y$ of the experiment in consideration by adding more patterns. This can be remedied by noting that the Eberhard inequality   is equivalent to the CH-inequality \cite{Clauser1974} of Eq.~\eqref{eQUATION:CH-inequality} if the null output is assigned to the outcome $1$ for every input. Indeed, suppose this assignment is done,  so the three outcome behaviour containing the probabilities $\tilde\wp(ab|xy)$ maps to the two-outcome behaviour with probabilities $\wp'(ab|xy)$ defined by 
\begin{align}
    \wp'(ab|xy) = \begin{cases}
        \tilde\wp(00|xy) & \textrm{if } ab=00\\
        \tilde \wp(01|xy)+ \tilde \wp(0 \emptyset|xy)& \textrm{if } ab=01\\
        \tilde\wp(10|xy) +\tilde\wp(\emptyset 1|xy)& \textrm{if } ab=10\\
        \tilde \wp(11|xy) + \tilde{\wp}(\emptyset \emptyset|xy) &\textrm{if } ab=11
    \end{cases}
\end{align}
for all $xy$. 
The distributions $\wp'(ab|xy)$ obey no-signalling, which  follows from the no-signalling constraints $\sum_{a\in \{\emptyset, 0,1 \}}\tilde{\wp}(a0|x 1)= \tilde{\wp}(b=0|y=1)$ and $\sum_{b\in \{\emptyset, 0, 1\}} \tilde{\wp}(0b|1y) = \tilde{\wp}(a=0|x=1)$. Under this mapping, if the behaviour consisting of $\tilde{\wp}(ab|xy)$ obeys the Eberhard inequality of $\eqref{Equation:Eberhard'sInequality}$ then the behaviour with $\wp'(ab|xy)$ defined equivalently obeys the inequality 
\begin{align}
\wp'(00|11)+ \wp'(10|10) + \wp'(01|01) - \wp'(00|00) \geq 0.
\end{align}
By using $\wp'(01|01) = \wp'(a=0|x=0) - \wp'(00|01)$ and $\wp'(10|10)=\wp'(b=0|y=0) - \wp'(00|10)$ this is seen to be equivalent to 
\begin{align}
    \wp'(00|11) - \wp'(00|01) - \wp'(00|10) - \wp'(00|00) +\wp'(a=0|x=0) +\wp'(b=0|y=0) \geq 0,
\end{align}
which is the CH-inequality of Eq.~\eqref{eQUATION:CH-inequality} evaluated with the distributions obtained by use of the assignment strategy. Hence, under this map, the Eberhard inequality is violated by $\tilde\wp(ab|xy)$ if and only if the CH inequality is violated by $\wp'(ab|xy)$.  Therefore the same critical detector efficiency $\epsilon > 2/3$ is valid in the simplest Bell scenario, which has complexity $W= A^XB^Y = 16$.  In general, the bounds for the  Eberhard \cite{Eberhard1993}, CH \cite{Clauser1974} (and also the CHSH \cite{Cla69}) inequalities may be different depending on how one deals with the null outcomes \cite{Czechlewski18}, namely, whether the discard or assignment strategy is chosen. In Ref.~\cite{Czechlewski18}, it was shown that when the \emph{assignment strategy} is used, as is the case in this argument, those inequalities remain essentially equivalent up to no-signalling and normalization and hence the threshold bounds match.

\section{Eberhard's bound is the minimum detector efficiency required for Bell-nonlocality demonstrations independently of quantum physics \label{SupplementalSection:BlackBoxEberhard}}

Interestingly, the bound $\epsilon = 2/3$ may, in fact, be identified as the minimum threshold detector efficiency required in arbitrary no-signalling theories for the demonstration of Bell-nonlocality in the simplest setup. For example, see Ref.~\cite{Massar2003}, where an explicit method to construct Local Hidden Variable models for efficiencies $\epsilon \leq 2/3$ is derived assuming no-signalling; or the proof of Theorem 6 in Ref.~\cite{Larsson2001}, which shows the necessity of $\epsilon >2/3$ for the violation of the Clauser-Horne inequality \cite{Clauser1974}  and also analytically constructs an appealing quantum model that approaches the bound in the appropriate limit. For completeness, we show a simple technique to derive this bound for the CHSH-inequality \cite{Cla69}, which is equivalent to the CH and Eberhard inequalities when the assignment strategy is used \cite{Czechlewski18}.

Let $P^{\rm PR}$ denote a behaviour, \emph{i.e.}~a collection of distributions $ P^{\rm PR}(ab|xy)$ with $x,y \in \{0, 1 \}$ and $a,b \in \{0, 1 \}$, of the  Popescu-Rohrlic (PR)-type \cite{Popescu1994}:
\begin{align}
    P^{\rm PR}(ab|xy) = \begin{cases}
        \frac{1}{2} & \textrm{ if } a=b \textrm{ and } xy \in \{00,01,10\} \vspace{1ex}  \\
        \frac{1}{2} & \textrm{ if } a\neq b \textrm{ and } xy=11 \;.
    \end{cases}
 \end{align}
The PR behaviour has the property that the CHSH-expression $C$
\begin{align}
  C = C_{00} + C_{01} + C_{10} - C_{11}, \label{CHSHinequality}
\end{align}
with $C_{xy} = P(00|xy) +P(11|xy) - P(10 |xy) - P(01|xy)$ reaches the algebraic upper bound $C =4$ when evaluated with respect to $P^{\rm PR}$, while every behaviour compatible with Local Causality \cite{Bell76} satisfies $C\leq 2$. 

When each detector, for each pair of measurement settings $x,y$, has a probability $\epsilon \in [0,1]$ of working as intended, and probability $1-\epsilon$ of providing a null outcome $\emptyset$ when a photon is incident, the behaviour $P^{\rm PR}$ can be thought to map to the 9-outcome behaviour defined by
\begin{align}
    P_{\epsilon,\emptyset}(ab|xy) = \begin{cases}
        \epsilon^2 P^{\rm PR}(a,b|x,y) & \\
        \epsilon (1-\epsilon) P^{\rm PR}(a|x) & \textrm{ if } b_y = \emptyset \\
        (1-\epsilon)\epsilon P^{\rm PR}(b |y) & \textrm{ if } a_x = \emptyset\\
        (1-\epsilon)^2 & \textrm{ if } a_x,b_y = \emptyset \label{9-outcomeBehaviour}\;.
    \end{cases}
\end{align}
This case can be mapped back to the two-outcome per input scenario by locally assigning the null outcomes to either of the other possible outcomes. We could use the strategy of assigning all null outcomes to 1 as in Section \ref{section:SM,EfficiencyIntheSimplestScenario}, however the more general case with arbitrary assignment strategies is treated just as easily.    The most general assignment strategies include \emph{local} mapping probabilities  $P_{x/y}(\emptyset \mapsto +1)$ and $P'_{x/y}(\emptyset \mapsto -1) = 1-P'_{x/y}(\emptyset \mapsto +1) $ :=$P'(a/b|x/y)$, by virtue of which Eq.~\eqref{9-outcomeBehaviour} collapses to
\begin{align}
    P_{\epsilon}(ab|xy) &= 
        \epsilon^2 P^{\rm PR}(ab|xy) + \epsilon(1-\epsilon) P^{\rm PR}(a|x)P'(b|y) 
        + \epsilon(1-\epsilon) P'(a|x)P^{\rm PR}(b|y) + (1-\epsilon)^2 P'(a|x)P'(b|y)\;.
\end{align}
 Since the CHSH inequality is linear in probabilities, plugging the distribution $P_{\epsilon}$ into it would essentially amount to a linear sum of the individual terms. The marginals $P^{\rm PR}(a|x), P^{\rm PR}(b|y)$ of the PR-behaviour are unbiased, and so the cross terms $\epsilon(1-\epsilon) P^{\rm PR}(a|x)P'(b|y) 
        $ and $ \epsilon(1-\epsilon) P'(a|x)P^{\rm PR}(b|y)$ vanish for the expression $C(P_{\epsilon})$. The last term proportional to $P'(a|x)P'(b|y)$, on the other hand, would evaluate to a number between $[-2,2]$ since it arises from the local post-processing. The \emph{worst case scenario} would be when the term assumes the value $+2$ since then both of the remaining terms contribute similarly to the violation of the CHSH-expression. From here, a necessary condition  for the violation is obtained as $
       \epsilon^2\times 4  + (1-\epsilon)^2  \times 2 >2  \Leftrightarrow \epsilon(3\epsilon - 2) > 0$ $\Leftrightarrow \epsilon > 2/3$, the Eberhard bound. 

       The above simple argument assumes that the behaviour in question is the PR-box, and hence may \emph{a priori} not be sufficient to (independently, anyway) guarantee that no value less than or equal to $\epsilon = 2/3$ is sufficient. Indeed, the PR-box has the `special' property of unbiased marginals. A simple observation using the properties of no-signalling correlations can be used to generalize this for the case of any no-signalling behaviours, including those reproducible in quantum mechanics. 

       In particular, the set of no-signalling correlations is a convex polytope, with 8 non-classical vertices corresponding to relabelings of the PR-behaviour \cite{Barret2005nonlocalCorrelations}. The PR-behaviours are in one-to-one correspondence with the CHSH-Bell inequalities, which form the facets of the Bell-local polytope, and so every behaviour $P^{\rm NS}(ab|xy)$ which violates a given CHSH-inequality, in particular the inequality $C \leq 2$ where $C$ is as in Eq.~\eqref{CHSHinequality}, may be expressed as the convex combination 
       \begin{align}
           P^{\rm NS}(ab|xy) = \sum_{i}(1 -\lambda_i')P^{\rm PR}(ab|xy) + \sum_i \lambda_i'P_i^{\rm L}(ab|xy) \equiv (1-\lambda) P^{\rm PR}(ab|xy) + \lambda P^{\rm L}(ab|xy) \;, \label{equation:LambdaDependentExpansion}
       \end{align}
where $P^{\rm L}_i(ab|xy)$ and $P^{\rm L}(ab|xy)$ are some (Bell-local) behaviours which do not violate the CHSH inequality. Running exactly the kind of argument as before now leads to an expression which depends on both $\lambda$ and $\epsilon$. Namely now

\begin{align}
    P_{\epsilon}^{\rm NS}(ab|xy) = \epsilon^2 P^{\rm NS}(ab|xy) + \epsilon(1-\epsilon)P^{\rm NS}(a|x)P'(b|y) + \epsilon(1-\epsilon)P^{\rm NS}(a|x)P'(b|y) + (1-\epsilon)^2P'(a|x)P'(b|y)\;.
\end{align}
The difference is that now the cross terms do not vanish, as the local marginals in the convex expansion \eqref{equation:LambdaDependentExpansion} may be unbiased. Furthermore, the first and last terms acquire $\lambda-$dependence. Other than that, the steps are similar. From the linearity of the CHSH expression and considering again the worst-case scenario where \emph{every} local distribution reaches the bound 2, the necessary condition for violation of $C(P_{\epsilon}^{\rm NS})>2$ is obtained as
\begin{align}
    \epsilon^2 \left[ (1-\lambda)\times 4 + \lambda \times 2\right] + 2\epsilon(1- \epsilon) \left[\lambda \times 2 \right] + (1-\epsilon)^2\times 2 > 2
\end{align}
which is equivalent to 
\begin{align}
    \epsilon(1-\lambda) \times \left[ 6\epsilon - 4 \right] >0 \Rightarrow \epsilon > \dfrac{2}{3}\;,
\end{align}
where the fact that $\epsilon(1-\lambda) >0$ for all $\epsilon, \lambda \in (0,1)$ was used. This shows that the Eberhard bound $\epsilon = 2/3$ is indeed the critical detection efficiency threshold for Bell-nonlocality demonstrations for all no-signalling behaviours in the simplest scenario, not just those reproducible in quantum physics. 

\section{Experimental details} 
\label{sm:experiment_details}

\subsection{Experimental setup}
\label{subsm:setup}
The source implemented for our experiment generates tunable polarization-entangled photon pairs of 1550 nm via type-II spontaneous parametric down conversion (SPDC) happening in a periodically poled potassium titanyl phosphate (PPKTP) crystal embedded in a beam-displacer interferometer, pumped with a continuous-wave laser of 775 nm \cite{Tis18, shalmStrongLoopholeFreeTest2015,villegas2024nonlocality}. The characterization of the entangled photons from the source is done by performing quantum state tomography \cite{whiteMeasuringTwoqubitGates2007} to reconstruct the density matrix of the entangled state obtained. The reconstructed state has a fidelity of $0.9953 \pm 0.0006$ with the maximally entangled state. The generated photons from the pair are sent to Alice and Bob.

The detection efficiency of Alice is quantified by the heralding efficiency, which is measured by dividing the sum of coincidences between Alice's and Bob's outcomes by the sum of singles obtained on Bob's side. The heralding efficiency is continuously monitored by calculating the sum of the heralding efficiencies of Alice's events recorded when measuring along the two orthogonal polarization axes: one axis corresponding to Alice's measurement setting, and the orthogonal axis obtained by rotating the half-wave plate (HWP) to $+45^{\circ}$ from that position. It is important to note that the coincidences from the orthogonal axis measurement are not used to calculate the steering inequality violation. 
The heralding efficiency of the source is optimized by adjusting the setup parameters, including pump and detection beam waists, and the efficiency of the superconducting nanowire single-photon detectors (SNSPDs) used for detection. To vary the detection efficiencies for the steering curves in Fig.~2 of the main text, we used the polarization dependence of the SNSPD efficiency.
Bob's measurement is realized as a POVM, with trine elements $E_0,E_1,E_2$ represented by the following unit vectors lying in the X-Z plane of the Bloch sphere:

\begin{subequations}\label{POVM}
    \begin{align}
    E_0 &= |V\rangle\\
    E_1 &= \frac{\sqrt{3}}{2}|H\rangle + \frac{1}{2}|V\rangle\\
    E_2 &= \frac{\sqrt{3}}{2}|H\rangle - \frac{1}{2}|V\rangle \;.
\end{align}
\end{subequations}
 $E_0$ is measured at the output 1, which is the reflecting port of the partially polarizing beam splitter (PPBS), with transmissivities $\tau_{V}=\sqrt{1/3}$, $\tau_{H} = 1$, and reflectivities $r_{V}=\sqrt{2/3}$, $r_{H} = 0$ for vertically and horizontally polarized light, respectively. A HWP at $22.5^\circ$ rotates the photons coming through the transmitted arm of the PPBS, and $E_1$  and $E_2$ are implemented at the reflected (output 2) and transmitted side (output 3) of the PBS as in Fig.~1 of the main text. 
 
\subsection{Evaluating the optimal steering inequality from data}
\label{sm:MLE}

To calculate the value of steering inequalities, the coincidences between Alice's and Bob's detections are measured with Bob's photons as heralding photons. The used inequalities are optimized from the corresponding assemblages of Bob, determined by a Maximum Likelihood Estimation (MLE) \cite{banaszek1999maximum, villegas2025quantum} on the observed measurement, guaranteeing the assemblage is no-signalling.

The experimental data obtained are in the form of singles and coincidence counts. 
Singles counts for each detector represent the total number of photons detected by them during the integration window, which is the time duration during which data is collected when the apparatus remains in a single measurement setting (3 seconds in our experiment). 
The coincidence counts between Alice and each of Bob's three outcomes are the instances where Alice detects a photon simultaneously with Bob's corresponding detector. 
The probabilities used to calculate the steering parameters include Alice's marginals and Bob's conditionals, represented by $P(a|x)$ and $P(b|a,x)$, respectively. 
The marginals represent the probability that Alice obtains an outcome ($a=+$ or $a=\emptyset$) given a particular setting ($x=0$ or $x=1$).  
Table \ref{tab:probabilities} shows how we calculate the probability values from the raw experimental data for a measurement setting x. 

\begin{table}[h]
\centering
\renewcommand{\arraystretch}{2.3} 
\setlength{\tabcolsep}{12pt} 
\begin{tabular}{|c|c|}
\hline
\textbf{Probabilities} & \textbf{Calculation of probability from singles and coincidences} \\ [6pt]
\hline
$P(a=+|x)$ &
{\small $\displaystyle \frac{C(A,B_1)+C(A,B_2)+C(A,B_3)}{S(B_1)+S(B_2)+S(B_3)}$} \\ [6pt]
\hline
$P(a=\emptyset|x)$ &
{\small $\displaystyle 1 - \frac{C(A,B_1)+C(A,B_2)+C(A,B_3)}{S(B_1)+S(B_2)+S(B_3)}$} \\ [6pt]
\hline
$P(b_1|a=+,x)$ &
{\small $\displaystyle \frac{C(A,B_1)}{C(A,B_1)+C(A,B_2)+C(A,B_3)}$} \\ [6pt]
\hline
$P(b_2|a=+,x)$ &
{\small $\displaystyle \frac{C(A,B_2)}{C(A,B_1)+C(A,B_2)+C(A,B_3)}$} \\ [6pt]
\hline
$P(b_3|a=+,x)$ &
{\small $\displaystyle \frac{C(A,B_3)}{C(A,B_1)+C(A,B_2)+C(A,B_3)}$} \\ [6pt]
\hline
$P(b_1|a=\emptyset,x)$ &
{\small $\displaystyle \frac{S(B_1)-C(A,B_1)}{S(B_1)+S(B_2)+S(B_3) - [C(A,B_1)+C(A,B_2)+C(A,B_3)]}$} \\ [6pt]
\hline
$P(b_2|a=\emptyset,x)$ &
{\small $\displaystyle \frac{S(B_2)-C(A,B_2)}{S(B_1)+S(B_2)+S(B_3) - [C(A,B_1)+C(A,B_2)+C(A,B_3)]}$} \\ [6pt]
\hline
$P(b_3|a=\emptyset,x)$ &
{\small $\displaystyle \frac{S(B_3)-C(A,B_3)}{S(B_1)+S(B_2)+S(B_3) - [C(A,B_1)+C(A,B_2)+C(A,B_3)]}$} \\ [6pt]
\hline
\end{tabular}
\caption{Probability calculation from raw data. $C(A, B_1)$, $C(A, B_2)$, and $C(A, B_3)$ represent coincidences between Alice's and Bob's outcomes. $S(A)$, $S(B_1)$, $S(B_2)$, and $S(B_3)$ represent the singles counts of Alice, and of the first, second, and third detectors of Bob, respectively.}
\label{tab:probabilities}
\end{table}

We test for violation of a steering inequality using the following method.
From the experimental events, we compute all the probabilities in Table~\ref{tab:probabilities}.
From these, we construct a \emph{most-likely} candidate assemblage $\{\sigma_{a|x}\}_{a,x}$ that is consistent with the non-signalling constraint $\sum_a \sigma_{a|x} = \rho_B,~\forall \  x$. 
Importantly, only the probabilities appearing in Table~\ref{tab:probabilities}, and a tomographically characterized POVM by Bob ${E_b}_{b=0}^2$ are utilized---no assumption is made about Alice's measurement device.
Using $p(a,b|x) = p(a|x)p(b|a,x)$, we construct the likelihood function for observing the raw data from ostensibly measuring the ensembles $\{\sigma_{a|x}\}$ with Bob's POVM as $\mathcal{L} = \prod _{a,b,x} \Tr[E_b \sigma_{a|x}]^{p(a,b|x)}$.
Taking a logarithm of $\mathcal{L}$ allows computation of a maximum likelihood candidate assemblage via a single semidefinite program instance:
\begin{equation}
\begin{aligned}
\quad & \text{max} & & \sum\limits_{a,b,x} p(a,b|x)~\log \Tr [E_b\sigma_{a|x}] &\\
& \text{s.~t.~} & & \sigma_{a|x} \geq 0 ~&\forall \ ~a,x~\\
& & & \sum\limits_a \sigma_{a|x} = \rho_B~ \forall \ ~x & \\
& & & \Tr[\rho_B] = 1 \;.
\end{aligned}
\label{eq:MLE}
\end{equation}
Each run, we compute these SDPs using the MOSEK solver~\cite{mosek}, and provide our Python implementation in the repository accessible at Ref.~\cite{code_repo}.
The optimization variables $\{\sigma_{a|x}\}_{a,x}$ and $\rho_B$ are passed to the steering inequality in Eq.~\eqref{eq:SE_val_exact}.

\subsection{Experimental steering with non-maximally entangled states} 
\label{subsm:nonmax}
We study how a non-maximally entangled state behaves in the simplest steering scenario. For that, we generate a state with $\alpha = 0.31~$radians  in 
Eq.~(5) of the main text,
through rotating the pump HWP  before the first beam displacer (BD) in the source 
(see Fig.~1 in the main text).
This state has a concurrence value of 0.4399, much lower than the unit concurrence for a maximally entangled state, indicating a lower entanglement in the state. The fidelity of the experimentally generated non-maximally entangled state with the closest ideal non-maximally entangled state in Eq.~(5) (with the same $\alpha$ value as in the main text) is
$0.9928 \pm 0.0002$. Fig.~\ref{unbalanced} shows that the generated non-maximally entangled state also shows negative steering parameters, indicating its steerability.

\begin{figure}[h]
\begin{center}
\includegraphics[width=0.65\textwidth]{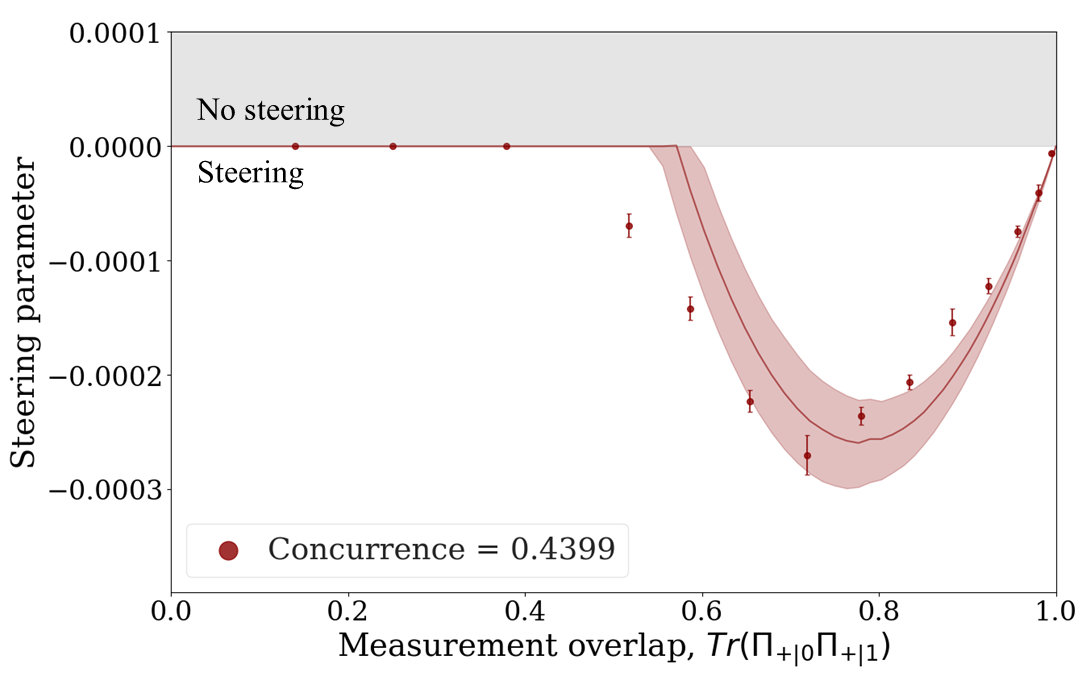}
\caption{\footnotesize 
Steering parameters versus Alice's measurement overlap for non-maximally entangled states. 
Points in the white region indicate a steering violation. 
The shaded regions of the curves represent theoretical predictions for an ideal state and POVM, based on the experimentally measured efficiency.
The band reflects $\pm1$ standard deviation uncertainty in efficiency. 
Markers represent the experimental data, along with the associated error estimated by repeating the measurements 10 times.}
\label{unbalanced}
\end{center}
\end{figure}

\subsection{Steering parameter values for maximally entangled states}
\label{subsm:steeringvalues}
Here, we provide the numerical values of the experimentally obtained  steering parameters in Fig.~2 of the main text. The lowest point on each curve is where the violations are highest, and the overlap values are near optimal. Thus, we provide the numerical values of the minimum steering parameters for all efficiencies along with their uncertainties in Table~\ref{tab:violation  values}. The uncertainties indicate the error of the mean value of the steering parameter calculated from 10 iterations of the experiment.

\begin{table} [H]
\renewcommand{\arraystretch}{2.3}
\setlength{\tabcolsep}{12pt}  
    \centering
    \begin{tabular}{|c|c|c|}
        \hline
        \textbf{Efficiency} & \textbf{Minimum steering parameter} & \textbf{Uncertainty in steering parameter}\\
        \hline
         $0.615 \pm 0.004$ & $-3.3 \times 10^{-3}$  & $0.3 \times 10^{-3}$ \\
         \hline
         $0.578 \pm 0.004$ & $-1.3 \times 10^{-3}$ & $0.1 \times 10^{-3}$ \\
         \hline
         $0.544 \pm 0.003$ & $-3.7 \times 10^{-4}$ & $0.4 \times 10^{-4}$ \\
         \hline
         $0.516 \pm 0.004$ & $-7.8 \times 10^{-5}$ & $0.4 \times 10^{-5}$  \\
         \hline
    \end{tabular}
    \caption{The values of the most negative steering parameter observed for each of the efficiencies included in Fig.~2 of the main text.}
    \label{tab:violation  values}
\end{table}

\subsection{Estimating noise robustness for experimental violations}
\label{sec:noisy_assemblage_witness}

We now investigate how noise robust our steering inequality violations are in the simplest scenario.
This computation results in the four data points illustrated in Fig.~3 of the main text, which witness steering approaching $\epsilon \rightarrow 1/2$.
To compare with the robustness of Eberhard's inequality, we estimate the fraction of white noise $\eta$ permissible in the underlying state, while maintaining steering inequality violation.
Mathematically, this corresponds to Alice performing her measurements on the state 
\begin{equation}
(1-\eta)\ketbra{\Psi_{AB}} + \eta\frac{I_A}{d_A}\otimes \frac{I_B}{d_B}.
\end{equation}
At the level of the assemblage produced for Bob, this modifies the assemblage according to
\begin{equation}
\sigma_{a|x} \rightarrow (1-\eta) \sigma_{a|x} + \eta\Tr[E_{a|x}] \frac{I}{d_A d_B}.
\end{equation}
We can compute the WNR $\eta^\star$ by the semidefinite program:
\begin{equation}
\begin{aligned}
\quad & \text{min} & & \eta &\\
& \text{s.~t.~} & & \sum_\lambda D(a|x,\lambda)\sigma_\lambda = (1-\eta) \eta\sigma_{a|x} + \eta \frac{\Tr[E_{a|x}]}{d_A} \frac{I}{d_B}& \forall \  a,~x,\\
& & & \sigma_\lambda \geq 0 & \forall \  \lambda\;.
\end{aligned}
\label{eq:SM_eig_SDPwnr}
\end{equation}
This procedure gives the four experimental data points appearing in Fig.~3 of the main text.

\end{document}